\documentclass[12pt]{article}
\usepackage{amsfonts}
\usepackage{amsmath}
\usepackage{amssymb}
\usepackage{amsthm}
\usepackage{setspace}
\usepackage{fullpage}
\usepackage{graphicx}
\usepackage{url}
\usepackage{verbatim}
\usepackage{enumerate}

\usepackage[pdftex,plainpages=false,hypertexnames=false,pdfpagelabels]{hyperref}
\usepackage{xcolor}
\definecolor{medium-blue}{rgb}{0,0,.8}
\hypersetup{
   colorlinks, linkcolor={black},
   citecolor={medium-blue}, urlcolor={medium-blue}
}

\usepackage{tikz}
\usetikzlibrary{arrows}
\usetikzlibrary{calc}

  \newcommand{\tikzmath}[2][]
     {\vcenter{\hbox{\begin{tikzpicture}[#1]#2
                     \end{tikzpicture}}}
     }

\newcommand{\nc}[2]{\newcommand{#1}{#2}}
\nc{\C}{\mathbb{C}}
\nc{\R}{\mathbb{R}}
\nc{\Q}{\mathbb{Q}}
\nc{\Z}{\mathbb{Z}}
\nc{\N}{\mathbb{N}}
\nc{\cA}{\mathcal{A}}
\nc{\cB}{\mathcal{B}}

\nc{\g}{\mathfrak{g}}

\newcommand{\DeclareMyOperator}[1]{ \expandafter\DeclareMathOperator\csname #1\endcsname{#1} }

\def\Rep{\mathrm{Rep}}

\def\Diff{\mathrm{Diff}}

\begin{document}
\title{Loop groups and diffeomorphism groups of the circle as colimits}
\author{\sc Andr\'e Henriques}
\date{\small \it University of Oxford, Mathematical Institute
\\
\sf andre.henriques@maths.ox.ac.uk}
\begin{singlespace}
\maketitle
\end{singlespace}

\newtheorem{theorem}{Theorem}
\newtheorem*{theorem*}{Theorem}
\newtheorem*{maintheorem}{Main Theorem}
\newtheorem{lemma}[theorem]{Lemma}
\newtheorem*{tech-lemma}{Technical lemma}
\newtheorem{proposition}[theorem]{Proposition}
\newtheorem{corollary}[theorem]{Corollary}
\newtheorem{definition}[theorem]{Definition}
\newtheorem{observation}[theorem]{Observation}
\newtheorem{exercise}[theorem]{Exercise}
\newtheorem{conjecture}[theorem]{Conjecture}
\newtheorem{expfact}[theorem]{Expected Fact}
\newtheorem*{conjecture*}{Conjecture}
\newtheorem{question}[theorem]{Question}
\newtheorem*{question*}{Question}

\newenvironment{maindefinition}[1][Main Definition.]{\begin{trivlist}
\item[\hskip \labelsep {\bfseries #1}]}{\end{trivlist}}

\theoremstyle{remark}
\newtheorem{remark}[theorem]{Remark}
\newtheorem*{remark*}{Remark}
\newtheorem{example}[theorem]{Example}
\newtheorem*{example*}{Example}

\abstract{In this paper, we show that loop groups and the universal cover of $\Diff_+(S^1)$ 
can be expressed as colimits of groups of loops/diffeomorphisms supported in subintervals of $S^1$.
Analogous results hold for based loop groups and for the based diffeomorphism group of $S^1$.
These results continue to hold for the corresponding centrally extended groups.

We use the above results to construct a comparison functor from the representations of a loop group conformal net to the representations of the corresponding affine Lie algebra.
We also establish an equivalence of categories between solitonic representations of the loop group conformal net, and locally normal representations of the based loop group.
}

\tableofcontents

\section{Introduction and statement of results}
In the category of groups, the concept of \emph{colimit} is a simultaneous generalisation of the notions of direct limit, and amalgamted free product.
Given a diagram of groups $\{G_i\}_{i\in \mathcal I}$ indexed by some poset $\mathcal I$ (i.e., a functor from $\mathcal I$ viewed as a category into the category of groups)
the colimit $\mathrm{colim}_{\mathcal I}\, G_i$ is the quotient
\[
\mathrm{colim}_{\mathcal I}\, G_i=\big(\ast_{i\in \mathcal I} G_i\big)/N
\]
of the free product of the $G_i$ by the normal subgroup $N$ generated by the elements $g^{-1} f_{ij}(g)$ for $g\in G_i$,
where $f_{ij}:G_i\to G_j$ are the homomorphisms in the diagram.
If the diagram takes values in the category of topological groups
(i.e., if the $G_i$ are topological groups and the $f_{ij}$ are continuous), then we may take the colimit in the category of topological groups.
The underlying group remains the same, but it is now endowed with the \emph{colimit topology}: the finest group topology such that all the maps $G_i\to G$ are continuous.

In the present paper, we show that loop groups and the universal cover of $\Diff_+(S^1)$ 
can be expressed as colimits of groups of loops/diffeomorphisms supported in subintervals\footnote{The the poset of subintervals of $S^1$ is not directed, so the colimits in question are not direct limits.} of $S^1$.

Let now $G$ be a compact Lie group, and 
let $LG:=\mathrm{Map}(S^1,G)$ be the group of smooth maps of $S^1:=\{z\in\C:|z|=1\}$ into $G$. This is the so-called \emph{loop group} of $G$.
For every interval $I\subset S^1$, let $L_IG\subset LG$ denote the subgroup of loops whose support is contained
in $I$.
If $G$ is simple and simply connected, then $LG$ admits a well-known central extension by $U(1)$ \cite{MR733047, MR1032521, MR900587}:
\begin{equation}\label{eq: SES LG}
\tikzmath[xscale=1.5]{
\node (0) at (0,0) {$0$};
\node (1) at (1,0) {$U(1)$};
\node (2) at (2.05,.05) {$\widetilde{LG}$};
\node (3) at (3,0) {$LG$};
\node (4) at (4,0) {$0$.};
\draw [->] (0) -- (1);
\draw [->] (1) -- ++(.75,0);
\draw [<-] (3) -- ++(-.65,0);
\draw [->] (3) -- (4);
}
\end{equation}
Letting $\widetilde{L_IG}$ be the restrictions of that central extension to the subgroups $L_IG$,
our main result about loop groups is that
\[
LG\,=\,\underset{I\subset S^1}{\mathrm{colim}}\,\, L_IG
\qquad\text{and}\qquad
\widetilde{LG}\,=\,\underset{I\subset S^1}{\mathrm{colim}}\,\, \widetilde{L_IG},
\]
where the colimit is taken in the category of topological groups.

Let $\Omega G$ be the subgroup of $LG$ consisting of loops that map the base point $1\in S^1$ to $e\in G$, and all of whose derivatives vanish at that point. 
We call $\Omega G$ the \emph{based loop group} of $G$ (this is the version of the based loop group that was used in \cite{CS(pt)}).
Letting $\widetilde{\Omega G}$ be the central extension of $\Omega G$ induced by \eqref{eq: SES LG},
we prove that
\[
\Omega G\,=\underset{I\subset S^1,\,1\;\!\not\in\;\!\mathring I\,\,\,}{\mathrm{colim^{\scriptscriptstyle H}}} L_IG
\quad\qquad\text{and}\qquad\quad
\widetilde{\Omega G}\,=\underset{I\subset S^1,\,1\;\!\not\in\;\!\mathring I\,\,\,}{\mathrm{colim^{\scriptscriptstyle H}}} \widetilde{L_IG},
\]
where 
$\mathring I$ denotes the interior of an interval $I$, and
$\mathrm{colim^{\scriptscriptstyle H}}\hspace{.3mm} \mathcal G_i$ denotes the maximal Hausdorff quotient of $\mathrm{colim}\, \mathcal G_i$, equivalently, the colimit in the category of Hausdorff topological groups.

Let $\Diff_+(S^1)$ be the group of orientation preserving diffeomorphisms of $S^1$, and let $\tilde{\Diff_+}(S^1)$ be its universal cover (with center $\Z$).
Given a subinterval $I\subset S^1$ of the circle, we write $\Diff_0(I)$ for the subgroup of diffeomorphisms that fix the complement of $I$ pointwise.
The groups $\Diff_0(I)$ are contractible and may therefore be treated as subgroups of $\tilde{\Diff_+}(S^1)$.
The group $\Diff_+(S^1)$ admits a well-known central extension by the reals, called the Virasoro-Bott group \cite{MR0488080, MR2456522,MR1674631}.
We write $\Diff_+^\R(S^1)$ for the Virasoro-Bott group and $\Diff_+^{\R\times \Z}(S)$ for its universal cover.
We then have the following system of central extensions:
\[\tikzmath[xscale=3]{
\node (31) at (3,1) {$\Diff_+(S^1)$};
\node (21) at (2,1) {$\Diff_+^\R(S^1)$};
\node (32) at (3,2) {$\tilde{\Diff_+}(S^1)$};
\node (22) at (2,2) {$\Diff_+^{\R\times\Z}(S^1)$};
\node (11) at (1,1) {$\R$};
\node (12) at (1,2) {$\R$};
\node (23) at (2,3) {$\Z$};
\node (33) at (3,3) {$\Z$};
\node (13) at (1,3) {$\R\times\Z$};
\node (41) at (4,1) {$0$};
\node (01) at (0,1) {$0$};
\node (40) at (4,0) {$0$};
\node (30) at (3,0) {$0$};
\node (20) at (2,0) {$0$};
\node (04) at (0,4) {$0$};
\node (02) at (0,2) {$0$};
\node (24) at (2,4) {$0$};
\node (34) at (3,4) {$0$};
\node (42) at (4,2) {$0$};
\draw [->] (01) -- (11);
\draw [->] (11) -- (21);
\draw [->] (21) -- (31);
\draw [->] (31) -- (41);
\draw [->] (02) -- (12);
\draw [->] (12) -- (22);
\draw [->] (22) -- (32);
\draw [->] (32) -- (42);
\draw [<-] (20) -- (21);
\draw [<-] (21) -- (22);
\draw [<-] (22) -- (23);
\draw [<-] (23) -- (24);
\draw [<-] (30) -- (31);
\draw [<-] (31) -- (32);
\draw [<-] (32) -- (33);
\draw [<-] (33) -- (34);
\draw [->] (04) -- (13);
\draw [->] (13) -- (22);
\draw [->] (22) -- (31);
\draw [->] (31) -- (40);
}\]
At last, let us write $\Diff_0^\R(I)$ for the restriction of the central extension by $\R$ to the subgroups $\Diff_0(I)\subset \Diff_+(S^1)$.
Our main result about diffeomorphism groups is that
\[
\tilde{\Diff_+}(S^1)\,=\,\underset{I\subset S^1}{\mathrm{colim}}\,\, \Diff_0(I)
\qquad
\text{and}
\qquad
\Diff_+^{\R\times \Z}(S^1)\,=\,\underset{I\subset S^1}{\mathrm{colim}}\,\, \Diff_0^\R(I).
\]

Let $\Diff_*(S^1)$ be the subgroup of $\Diff(S^1)$ consisting of diffeomorphisms that fix the point $1\in S^1$, and are tangent up to infinite order to the identity map at that point. 
We call it the \emph{based diffeomorphism group} of $S^1$.
Let $\Diff_*^\R(S^1)$ be the restriction of the central extension by $\R$ to $\Diff_*(S^1)$.
We also prove that
\[
\Diff_*(S^1)=\underset{\substack{I\subset S^1,\,1\;\!\not\in\;\!\mathring I\,\,\,}}{\mathrm{colim^{\scriptscriptstyle H}}}\, \Diff_0(I),
\quad\,\,\,
\Diff_*^\R(S^1)=\underset{\substack{I\subset S^1,\,1\;\!\not\in\;\!\mathring I\,\,\,}}{\mathrm{colim^{\scriptscriptstyle H}}}\, \Diff_0^\R(I).
\]

\begin{remark*}
All our results are formulated for the $C^\infty$ topology (uniform convergence of all derivates),
but they hold equally well for groups of $C^r$ loops $S^1\to G$, $r\ge 1$, and for groups of $C^r$ diffeomorphisms of $S^1$, $r\ge 2$ (with the exception of Section~\ref{sec: Solitons and representations of OmG}, which seems to requires $r\ge 4$ \cite{arXiv:1808.02384}).
\end{remark*}

In the third section of this paper, we apply the above results about (central extensions of) $LG$ and $\Diff(S^1)$ to the representations of loop group conformal nets.
For each compact simple Lie group $G$ and integer $k\ge 1$, there is a conformal net called the \emph{loop group conformal net}.
It associates to an interval $I\subset S^1$ a von Neumann algebra version of the twisted group algebra of $L_IG$.

We construct a comparison functor from the representations of a loop group conformal net to the representations of the corresponding affine Lie algebra.
We also establish an equivalence of categories between solitonic representations of the loop group conformal net, and locally normal representations of the based loop group.
This last result is needed in order to fill a small gap between the statement of \cite[Thm.\,1.1]{Bicommutant-categories-from-conformal-nets} and the results announced in \cite[\S4 and \S5]{CS(pt)}.
\bigskip\medskip

\newpage\noindent
{\large \bf Acknowledgements}\medskip\\
We thank Sebastiano Carpi for many useful discussions and references.
This research was supported by the ERC grant No 674978 under the European Union's Horizon 2020 research innovation programme.

\section{Colimits and central extensions}

The category of topological groups admits all colimits.
If $\{G_i\}_{i\in \mathcal I}$ is a diagram of topological groups indexed by some small category $\mathcal I$, then the colimit $G=\mathrm{colim}_{\mathcal I}\, G_i$ can be computed as follows.
As an abstract group, it is given by the colimit of the $G_i$ in the category of groups.
The topology on $G$ is the finest group topology that makes all the maps $G_i\to G$ continuous.

Let us call a small category $\mathcal I$ \emph{connected} if any two objects $i,j\in\mathcal I$ are related by a zig-zag $i\leftarrow i_1\to i_2\leftarrow i_3\to i_4\,\,\ldots\,\, i_n\leftarrow j$ of morphisms.
Similarly, let us call a diagram connected (`diagram' is just an other name for `functor') if the indexing category $\mathcal I$ is connected.

We recall the notion of central extension of topological groups:

\begin{definition}
Let $G$ be a topological group, and let $Z$ be an abelian topological group.
A central extension of $G$ by $Z$ is a short exact sequence
\[
0\,\longrightarrow\, Z\,\stackrel{\iota}\longrightarrow\, \tilde G \,\stackrel{\pi}\longrightarrow\, G \,\longrightarrow\, 0
\]
such that $Z$ sits centrally in $G$, the map $\iota$ is an embedding ($Z$ is equipped with the subspace topology)
and there exist a continuous local section $s:U\to \tilde G$ of $\pi$, where $U$ is an open neighbourhood of $e\in G$.
\end{definition}

Note that for $U$ and $s$ as above, the map $s\cdot \iota:U\times Z\to \tilde G$ is always an open embedding.

\begin{proposition}\label{prop: central ex of colim}
Let $\{G_i\}_{i\in \mathcal I}$ be a connected diagram of topological groups, let
\[
G:=\mathrm{colim}_{\mathcal I}\, G_i,
\]
and let $\varphi_i:G_i\to G$ be the canonical homomorphisms.
Assume that there exists a neighbourhood of the identity $U\subset G$, and finitely many continuous maps $f_i:U\to G_i$, $f_i(e)=e$, such that
\[
\qquad\qquad {\prod}_{\;\!i}\; \varphi_i\circ f_i(x)=x\qquad\quad \forall x\in U
\]
(the $f_i$ are indexed over an ordered finite subset of the objects of $\mathcal I$).

Let $0\to Z\to\widetilde G\to G\to 0$ be a central extension of $G$, and
let $\widetilde G_i$ be the induced central extensions of $G_i$ (the pullback of $\tilde G$ along the map $\varphi_i:G_i\to G$).
Then the canonical map
\[
\mathrm{colim}_{\mathcal I}\, \widetilde G_i \,\to\, \widetilde G
\]
is an isomorphism of topological groups.\footnote{
For $i\to j$ a morphism in $\mathcal I$, the map $\tilde G_i\to \tilde G_j$ is the unique group homomorphism which makes the diagrams
$\begin{smallmatrix}
\tilde G_i &\;\to\;& \tilde G_j \\
\downarrow &&\downarrow \\
G_i&\;\to\;& G_j
\end{smallmatrix}$
and
$\begin{smallmatrix}
&& \!\!\tilde G_j\!\! && \\
&\nearrow && \searrow&\\
\tilde G_i\!\! &&\!\!\longrightarrow\!\!&& \!\!\tilde G \\
\end{smallmatrix}$
commute (by the universal property of the pullback defining $\tilde G_j$).
}
\end{proposition}

\begin{proof}
The diagram is connected, so all the central $Z$'s in the various $\widetilde G_i$ get identified in $\mathrm{colim}_{\mathcal I}\, \widetilde G_i$.
Moreover, the canonical map $\iota_Z:Z\to \mathrm{colim}_{\mathcal I}\, \widetilde G_i$ is an embedding, because the triangle
\[
\tikzmath[scale=.8]{
\node (A) at (-3.3,0) {$Z$};
\node[inner xsep=4] (B) at (0,.8) {$\mathrm{colim}_{\mathcal I}\, \widetilde G_i$};
\node (C) at (3.3,0) {$\widetilde G$};
\draw[<-] (B.west)+(0,-.2) --node[above, pos=.45]{$\scriptstyle \iota_Z$} (A);
\draw[->] (B.east)+(0,-.2) -- (C);
\draw[->] (A.east)+(0,-.08) coordinate(x) -- (C.west|-x);
}
\]
commutes.
The quotient is easily computed:\vspace{-.1cm}
\[
\begin{split}
\big(\mathrm{colim}_{\mathcal I}\, \widetilde G_i\big)/Z
&= 
\mathrm{colim} \big(Z \tikzmath{
\node at (0,0) {$\longrightarrow$}; 
\node at (0,.2) {$\longrightarrow$};
\node at (0,.42) {$\scriptstyle \iota_Z$};
\node at (0,-.2) {$\scriptstyle 0$};}
(\mathrm{colim}_{\mathcal I}\, \widetilde G_i)\big)
\\
&=
\mathrm{colim}_{\mathcal I}\big(
\mathrm{colim} (Z \tikzmath{
\node at (0,0) {$\rightarrow$}; 
\node at (0,.2) {$\rightarrow$};}
\widetilde G_i
)\big)
=
\mathrm{colim}_{\mathcal I}\, G_i=G,
\vspace{-.1cm}
\end{split}
\]
so the sequence 
$
0\to Z\to \mathrm{colim}_{\mathcal I}\, \widetilde G_i \to G \to 0
$
is a short exact sequence of groups. 
To show that it is also a short exact sequence of \emph{topological} groups, we need to argue that the projection map $\mathrm{colim}_{\mathcal I}\, \widetilde G_i \to G$ admits local sections.
Pick neighbourhoods $U_i\subset G_i$ of the neutral elements and local sections $s_i:U_i\to \widetilde G_i$.
The map $x\mapsto \prod_i s_i(f_i(x))$, defined on the finite intersection $U:=\bigcap f_i^{-1}(U_i)\subset G$, is the desired local section.

We have two exact sequences of topological groups, and a map between them:
\[
\tikzmath[scale=.8]{
\node (Z) at (-3,1.6) {$0$};
\node (A) at (0,1.6) {$Z$};
\node[yshift=1] (B) at (4,1.6) {$\mathrm{colim}_{\mathcal I}\, \widetilde G_i$};
\node (C) at (8,1.6) {$\mathrm{colim}_{\mathcal I}\, G_i$};
\node (D) at (11,1.6) {$0$};
\node (Z') at (-3,0) {$0$};
\node (A') at (0,0) {$Z$};
\node (B') at (4,0) {$\widetilde G$};
\node (C') at (8,0) {$G$};
\node (D') at (11,0) {$0$};
\draw[->] (Z) -- (A);
\draw[->] (A) -- (B.west|-A);
\draw[->] (B.east|-A) -- (C);
\draw[->] (C) -- (D);
\draw[double, double distance=1] (A) -- (A');
\draw[->] (B) -- (B');
\draw[double, double distance=1] (C) -- (C');
\draw[->] (Z') -- (A');
\draw[->] (A') -- (B');
\draw[->] (B') -- (C');
\draw[->] (C') -- (D');
}
\]
By the five lemma, the middle vertical map is an isomorphism of groups.
It is an isomorphism of topological groups because both $\tilde G$ and $\mathrm{colim}_{\mathcal I}\, G_i$
are locally homeomorphic to a product $U\times Z$.
\end{proof}

\subsection{Loop groups}

We write $S$ for a manifold diffeomorphic to the standard circle $S^1:=\{z\in\C:|z|=1\}$, and call such a manifold \emph{a circle}.
We write $I$ for a manifold diffeomorphic to $[0,1]$, and call such a manifold \emph{an interval}.
All circles and intervals are oriented.

\subsubsection{The free loop group}\label{sec:The free loop group}
Let $G$ be a compact, simple, simply connected Lie group.
Throughout this section, we fix a circle $S$, and write $LG:=\mathrm{Map}(S,G)$ for the group of smooth maps from $S$ to $G$.
Given an interval $I$, we denote by $L_IG$ be the group of maps $I\to G$ that send the boundary of $I$ to the neutral element $e\in G$, and all of whose derivatives vanish at those points.
If $I$ is a subinterval of $S$, we identify $L_IG$ with the subgroup of $LG$ of loops with support in $I$.

\begin{theorem}\label{thm: LG=colimit}
For every circle $S$, the natural map
\[
{\mathrm{colim}}_{I\subset S}\, L_IG\,\,\to\,\,LG
\]
(colimit indexed over the poset of subintervals of $S$)
is an isomorphism of topological groups.
\end{theorem}

The group $LG=\mathrm{Map}(S,G)$ admits a well-known central extension $\widetilde{LG}$, constructed as follows.
Let $\g$ be the Lie algebra of $G$.
The Lie algebra $C^\infty(S,\g)$ of $LG$ admits a well-known $2$-cocycle given by the formula $(f,g)\mapsto \frac1{2\pi i}\int_S\langle f,dg\rangle$.
(Here, $\langle\,\,,\,\rangle:\g\times \g\to\R$ is the \emph{basic inner product} -- the smallest $G$-invariant inner product whose restriction to any $\mathfrak{su}(2)\subset \g$ is a positive integer multiple of the pairing $(X,Y)\mapsto -tr(XY)$.)
This cocycle can be used to construct a central extension of $C^\infty(S,\g)$ by the abelian Lie algebra $i\R$.
The latter can be then integrated to a simply connected infinite dimensional Lie group $\widetilde{LG}$ with center $U(1)$ \cite{MR733047, MR1032521, MR900587, MR1674631}.\vspace{-.7mm}

Let $\widetilde{LG}^k$ be the quotient of $\widetilde{LG}$ by the central subgroup $\mu_k\subset U(1)$ of $k$-th roots of unity.

\begin{theorem}\label{thm: LG=colimit -- BIS}
Let $S$ be a circle.
Given an interval $I\subset S$, let us denote by $\widetilde{L_IG}$ the pullback of the central extension $\widetilde{LG} \to LG$
along the inclusion $L_IG\to LG$.
Then the natural map
\[
{\mathrm{colim}}_{I\subset S}\, \widetilde{L_IG}\,\,\to\,\,\widetilde{LG}
\]
is an isomorphism of topological groups.

More generally, let $\widetilde{L_IG}^k\!$ be the pullback of $\widetilde{LG}^k\!$ along the inclusion $L_IG\to LG$.
Then the natural map
\[
\mathrm{colim}_{I\subset S^1}\widetilde{L_IG}^k\to \widetilde{LG}^k
\]
is an isomorphism of topological groups.
\end{theorem}

\begin{lemma}\label{lem: LG is genr'ted}
{\rm(i)}
Let $S$ be a circle and let $\mathcal I=\{I_i\}$ be a collection of subintervals whose interiors cover $S$.
Then the subgroups $L_{I_i}G$ generate $LG$.\\
{\rm(ii)}
Let $I$ be an interval and let $\mathcal I=\{I_i\}$ be a finite collection of subintervals whose interiors cover that of $I$.
Then the subgroups $L_{I_i}G$ generate $L_IG$.
\end{lemma}

\begin{proof}
The two statements are entirely analogous. We only prove the first one.
First of all, since $S$ is compact, we may assume without loss of generality that $n:=|\mathcal{I}|<\infty$.
Let $\exp:\g\to G$ be the exponential map, and let $u\subset \g$ be a convex neighbourhood of $0$ such that ${\exp}|_u:u\to G$ is a diffeomorphism onto its image $U:=\exp(u)$.

Given a partition of unity $\phi_i:I_i\to [0,1]$, every loop $\gamma\in \mathrm{Map}(S,U)\subset LG$ can be factored as
\begin{equation}\label{eq: part of 1 for LG}
\gamma=\prod_{i=1}^n\gamma_i
\end{equation}
with $\gamma_i(t)=\exp \big(\phi_i(t)\cdot \exp^{-1}(\gamma(t))\big)$.
(Note that the $\gamma_i$ commute.)
The subgroup generated by the $L_{I_i}G$ therefore contains $\mathrm{Map}(S,U)$.
The latter is open and hence generates $LG$.
\end{proof}

Let $\mathcal I$ be a collection of subintervals of $S$ whose interiors form a cover and that is closed under taking subintervals:
($I_1\in \mathcal I$ and $I_2\subset I_1$) $\Rightarrow$ $I_2\in \mathcal I$.
If $I_1$, $I_2 \in \mathcal I$ are such that $I_1\cap I_2$ is an interval, then the diagram 
\begin{equation}\label{eq: LG two ways}
\tikzmath{
\node (A) at (2.5,0) {$L_{I_1\cap I_2}G$};
\node (B) at (5,.8) {$L_{I_1}G$};
\node (C) at (5,-.8) {$L_{I_2}G$};
\node (D) at (7.6,0) {${\mathrm{colim}}_{\mathcal I}\,L_IG$};
\draw [->] (A) -- (B);
\draw [->] (A) -- (C);
\draw [<-] (D.north west) + (.2,-.03) --node[above, pos=.45]{$\scriptstyle \iota_1$} (B);
\draw [<-] (D.south west) + (.2,.03) --node[below, pos=.45]{$\scriptstyle \iota_2$} (C);
}
\end{equation}
clearly commutes.
When $I_1,I_2\in\mathcal I$ have disconnected intersection and $\gamma$ has support in $I_1\cap I_2$, it is not clear, a priori, that $\iota_1(\gamma)=\iota_2(\gamma)$.
Letting $J_1$, $J_2$ be the connected components of $I_1\cap I_2$, we can rewrite $\gamma$ as a product $\gamma=\gamma_1\,\gamma_2$, with $\gamma_i\in L_{J_i}G$.
We then have 
\[
\iota_1(\gamma)=\iota_1(\gamma_1\,\gamma_2)=\iota_1(\gamma_1)\iota_1(\gamma_2)=\iota_2(\gamma_1)\iota_2(\gamma_2)=\iota_2(\gamma_1\,\gamma_2)=\iota_2(\gamma).
\]
So the diagram \eqref{eq: LG two ways} always commutes, even when $I_1\cap I_2$ is disconnected.
Given $\gamma\in L_IG$ for some $I\in\mathcal I$, we also write $\gamma$ for its image in ${\mathrm{colim}}_{\mathcal I}\,L_IG$.
This element is well-defined by the commutativity of \eqref{eq: LG two ways}.

The following result is a strengthening of Theorem~\ref{thm: LG=colimit}:

\begin{theorem}\label{thm2: LG=colimit}
Let $\mathcal I$ be a collection of subintervals of $S$ whose interiors form a cover, and that is closed under taking subintervals.
Let $N\triangleleft\, {\mathrm{colim}}_{\mathcal I}\,L_IG$
be the normal subgroup generated by commutators of loops with disjoint supports:
\[
N:=\big\langle\,\gamma\,\delta\,\gamma^{-1}\delta^{-1}\,\big|\;\mathrm{supp}(\gamma),\mathrm{supp}(\delta)\in\mathcal I,\, \mathrm{supp}(\gamma)\cap\mathrm{supp}(\delta)=\emptyset\,\big\rangle.
\]
Then the natural map
\begin{equation}\label{eq: LG is a colimit+++}
\big({\mathrm{colim}}_{\mathcal I}\, L_IG\big)/N\,\to\,LG
\end{equation}
is an isomorphism of topological groups.
\end{theorem}

Before embarking in the proof, let us show how Theorems~\ref{thm: LG=colimit} and~\ref{thm: LG=colimit -- BIS} follow from the above result.

\begin{proof}[Proof of Theorem~\ref{thm: LG=colimit}]
Let $\mathcal I$ be the poset of all subintervals of $S$.
By Theorem \ref{thm2: LG=colimit}, the map $({\mathrm{colim}}_{I\subset S}\, L_IG)/N\,\to\,LG$ is an isomorphism.
So it suffices to show that $N$ is trivial.
Given two loops $\gamma_1,\gamma_2\in LG$ with disjoint support, let $I\subset S$ be an interval that contains the union of their supports.
The commutator of $\gamma_1$ and $\gamma_2$ is trivial in $L_IG$. It is therefore also trivial in the colimit.
\end{proof}

\begin{proof}[Proof of Theorem~\ref{thm: LG=colimit -- BIS}]
It is enough to show that the colimit which appears in Theorem~\ref{thm: LG=colimit} satisfies the assumption of Proposition~\ref{prop: central ex of colim}.
The maps $\gamma\mapsto \gamma_i$
used in equation \eqref{eq: part of 1 for LG}
provide the required factorization.
So Theorem~\ref{thm: LG=colimit} implies Theorem~\ref{thm: LG=colimit -- BIS}.
\end{proof}

\begin{proof}[Proof of Theorem \ref{thm2: LG=colimit}]
Given $\gamma\in L_IG$ for some $I\in\mathcal I$, we write $[\gamma]$ for its image in $({\mathrm{colim}}_{\mathcal I}\,L_IG)/N$.

Let $\{J_j\}_{j\,=\,1\ldots n}$ be a cover of $S$ such that each $J_j\cap J_{j+1}$ is an interval (in particular $J_j\cap J_{j+1}$ is non-empty) and the other intersections are empty (cyclic numbering).
The $J_j$ may be chosen small enough so that each union $J_{j-1}\cup J_{j}\cup J_{j+1}$ is in $\mathcal I$ (cyclic numbering).
Let $U\subset G$ be as in the proof of Lemma~\ref{lem: LG is genr'ted}.
By \eqref{eq: part of 1 for LG}, any loop $\gamma\in\mathrm{Map}(S,U)\subset LG$ can be factored as
$\gamma=\gamma_1\ldots  \gamma_n$, with $\gamma_j\in L_{J_j}G$.
Moreover, that factorisation may be chosen to depend continuously on~$\gamma$.
This provides a local section of the map in~\eqref{eq: LG is a colimit+++}:
\begin{equation*}
\tikzmath{
\node[below] (A) at (0,0) {$(\mathrm{colim}_{\mathcal I}\,L_IG)/N$}; 
\node[below] (B) at (4.2,-.02) {$LG$};
\draw[->] (A.east) -- (B.west|-A.east);
\node[inner sep=3] (C) at (2.5,-1.5) {$\mathrm{Map}(S,U)\,\,$};
\draw[<-right hook, shorten >=-.8, shorten <=.8] (B.south)+(-.3,0) -- (C);
\draw[dashed, ->, shorten <=2.5] (C) -- (1,-.75);
\node at (0+.1,-1) {$\scriptstyle [\gamma_1]\ldots [\gamma_n]$};
\node at                    (.83+.1,-1.48) {$\scriptstyle \gamma$};
\node[rotate=154] at (.5+.1,-1.3) {$\scriptstyle \mapsto$};
}
\end{equation*}
The map $(\mathrm{colim}_{\mathcal I}\,L_IG)/N\to LG$ is surjective by Lemma~\ref{lem: LG is genr'ted}.
Since there exists a continuous local section, all that remains to do in order to show that it is an isomorphism
is to prove injectivity.

Let $g$ be an element in the kernel.
By Lemma~\ref{lem: LG is genr'ted}, we may rewrite $g$ as a product
\(
[\gamma_1][\gamma_2]\ldots[\gamma_N]
\),
with $\gamma_i\in L_{J_i}G$ for $J_i\in\mathcal I$.
Since $g$ is in the kernel of the map to $LG$, the relation
\begin{equation}\label{eq: rel1 -- LG}
\gamma_1\, \gamma_2 \ldots \gamma_N=e
\end{equation}
holds in $LG$.

Any loop $\gamma\in\mathrm{Map}(S,U)$ can be factored as $\gamma=\gamma_1\ldots \gamma_n$ with $\gamma_j\in L_{J_j}G$.
The set $\mathcal U:=L_{J_1}G\cdot\ldots\cdot L_{J_n}G =\{\gamma_i\ldots \gamma_n\,|\,\gamma_j\in L_{J_j}G\}$ is therefore a neighbourhood of $e\in LG$.
Moreover, it is visibly path connected.
Since $\pi_2(G)=0$ \cite[\S 8.6]{MR900587}, we have
\[
\pi_1(LG)=\pi_1(\Omega G\times G)=\pi_2(G)\times \pi_1(G)=0.
\]
So, by Lemma~\ref{lem: vK diagr EZ}, 
equation \eqref{eq: rel1 -- LG} is a formal consequence of relations of length $3$ between the elements of $\mathcal U$.

\begin{lemma} \label{lem: vK diagr EZ}
Let $\mathcal G$ be a simply connected topological group, let $\mathcal U\subset \mathcal G$ be a path-connected neighbourhood of $e\in \mathcal G$, and
let $F$ be the free group on $\mathcal U$.
Then the kernel of the map $F\to \mathcal G$ is generated as a normal subgroup by words of length $3$.

\rm (In other words, any relation between elements of $\mathcal U$ is a formal consequence of relations of length $3$ between elements of~$\mathcal U$.)
\end{lemma}

\begin{proof}
Let $g_1^{\varepsilon_1}g_2^{\varepsilon_2}\ldots g_N^{\varepsilon_N}\in F$, $g_i\in \mathcal U$, $\varepsilon_i\in\{\pm1\}$, be a word in the kernel of the map $F\to \mathcal G$.
We want to show that the relation 
\begin{equation}\label{eq: few g's}
g_1^{\varepsilon_1}g_2^{\varepsilon_2}\ldots g_N^{\varepsilon_N}=e
\end{equation}
is a formal consequence of relations of length $3$ between elements of~$\mathcal U$.

For every $i$, pick a path $\gamma_i:[0,1]\to \mathcal U$ from $e$ to $g_i$.
Since $\mathcal G$ is simply connected, there exists a disk $D^2\to \mathcal G$ that bounds the loop
\begin{equation}\label{eq: lots of g's}
e \stackrel{\gamma_1^{\varepsilon_1}}{-\!-} g_1^{\varepsilon_1} \stackrel{g_1^{\varepsilon_1}\gamma_2^{\varepsilon_2}}{-\!-} g_1^{\varepsilon_1}g_2^{\varepsilon_2} \stackrel{g_1^{\varepsilon_1}g_2^{\varepsilon_2}\gamma_3^{\varepsilon_3}}{-\!-} g_1^{\varepsilon_1}g_2^{\varepsilon_2}g_3^{\varepsilon_3} \,\,{-\!-}\,\, \ldots\ldots \,\,{-\!-}\,\, g_1^{\varepsilon_1}g_2^{\varepsilon_2}\ldots g_N^{\varepsilon_N} = e.
\end{equation}
Triangulate $D^2$ finely enough and orient all the edges so that, for each oriented edge $x\,{-\!\!-}\;y$, the ratio $x^{-1}y$ is in $\mathcal U$.
The orientations along the boundary are chosen compatibly with the $\varepsilon_i$'s in \eqref{eq: few g's}.
Now forget the map $D^2\to \mathcal G$ and only remember the triangulation of $D$, along with the labelling of its vertices by elements of $\mathcal G$.

Before subdividing $D$, the word that one could read along the boundary of $D$ was $g_1^{\varepsilon_1}\ldots g_N^{\varepsilon_N}$.
After subdividing $D$, that word is now of the form
$(h_{11}h_{12}\ldots h_{1n_1})^{\varepsilon_1}$
$(h_{21}h_{22}\ldots h_{2n_2})^{\varepsilon_2}\ldots\ldots
(h_{N1}h_{N2}\ldots h_{Nn_N})^{\varepsilon_N}$\!\!\;,
with $h_{ij}\in \mathcal U$ and
\begin{equation}\label{eq: few h's}
h_{i1}h_{i2}\ldots h_{in_i}=g_i.
\end{equation}

Each little triangle $x-y-z-x$ of the triangulation corresponds to a $3$-term relation among elements of $\mathcal U$.
Depending on the orientation of the edges, this $3$-term relation could be any one of the following eight possibilities:
\[
\begin{matrix}
\scriptstyle(x^{\text{-}1}y)(y^{\text{-}1}z)(z^{\text{-}1}x)=e\,\,\,
&\scriptstyle(x^{\text{-}1}y)(y^{\text{-}1}z)(x^{\text{-}1}z)^{\text{-}1}=e\,\,\,
&\scriptstyle(x^{\text{-}1}y)(z^{\text{-}1}y)^{\text{-}1}(z^{\text{-}1}x)=e\,\,\,
&\scriptstyle(x^{\text{-}1}y)(z^{\text{-}1}y)^{\text{-}1}(x^{\text{-}1}z)^{\text{-}1}=e
\\
\scriptstyle(y^{\text{-}1}x)^{\text{-}1}(y^{\text{-}1}z)(z^{\text{-}1}x)=e\,\,\,
&\scriptstyle(y^{\text{-}1}x)^{\text{-}1}(y^{\text{-}1}z)(x^{\text{-}1}z)^{\text{-}1}=e\,\,\,
&\scriptstyle(y^{\text{-}1}x)^{\text{-}1}(z^{\text{-}1}y)^{\text{-}1}(z^{\text{-}1}x)=e\,\,\,
&\scriptstyle(y^{\text{-}1}x)^{\text{-}1}(z^{\text{-}1}y)^{\text{-}1}(x^{\text{-}1}z)^{\text{-}1}=e.
\end{matrix}
\]
The whole disc is a van Kampen diagram exhibiting the relation
\begin{equation}\label{eq: lots of h's}
(h_{11}h_{12}\ldots h_{1n_1})^{\varepsilon_1}
(h_{21}h_{22}\ldots h_{2n_2})^{\varepsilon_2}\ldots\ldots\ldots
(h_{N1}h_{N2}\ldots h_{Nn_N})^{\varepsilon_N}=e
\end{equation}
as a formal consequence of the above $3$-term relations
(see \cite[Chapt 4]{MR1191619} for generalities about van Kampen diagrams).

The relation \eqref{eq: few g's} is a formal consequence of the relations \eqref{eq: lots of h's} and \eqref{eq: few h's}.
Therefore, in order to finish the lemma, it remains to show that \eqref{eq: few h's} is a formal consequence of relations of length $3$ between elements of~$\mathcal U$.
By construction, $g_{ij}:=h_{i1}\ldots h_{ij}$ is in $\mathcal U$ for every $j\le n_i$.
Therefore
\[
g_{i,j}\,h_{i,j+1}=g_{i,j+1}
\]
is a 3-term relation between elements of $\mathcal U$.
One checks easily that \eqref{eq: few h's} is a formal consequence of the above 3-term relations.
\end{proof}

We have shown that equation \eqref{eq: rel1 -- LG} is a formal consequence of relations of length $3$ between elements of $\mathcal U=L_{J_1}G\cdot\ldots\cdot L_{J_n}G$.
It is therefore a formal consequence of certain relations
\begin{equation}\label{eq: rel2 -- LG}
(\delta_1\, \delta_2 \ldots \delta_{n})^{\varepsilon_1}(\delta_{n+1} \ldots \delta_{2n})^{\varepsilon_2}(\delta_{2n+1} \ldots \delta_{3n})^{\varepsilon_3}=e
\end{equation}
of length $3n$ between elements of the subgroup $L_{J_j}G$.
Here, $\delta_i$, $\delta_{n+i}$, $\delta_{2n+i}\in L_{J_i}G$.
The implication \eqref{eq: rel2 -- LG} $\Rightarrow$ \eqref{eq: rel1 -- LG} is formal:
any group generated by subgroups isomorphic to the $L_{J_j}G$'s in which the relations \eqref{eq: rel2 -- LG} hold also satisfies the relation \eqref{eq: rel1 -- LG}.

In order to prove that the equation $[\gamma_1][\gamma_2]\ldots[\gamma_N]=e$ holds,
it is therefore enough to show that the relations 
\begin{equation}\label{eq: lots of delta}
\big([\delta_1]\, [\delta_2] \ldots [\delta_n]\big)^{\varepsilon_1}\big([\delta_{n+1}] \ldots [\delta_{2n}]\big)^{\varepsilon_2}\big([\delta_{2n+1}] \ldots [\delta_{3n}]\big)^{\varepsilon_3}=e
\end{equation}
hold in $(\mathrm{colim}_{\mathcal I}L_IG)/N$.
Using that $[\delta_i]^{-1}=[\delta_i^{-1}]$,
we may rewrite \eqref{eq: lots of delta} as:
\begin{equation}\label{eq: our goal -- LG}
\prod_{i=1}^{3n}\,[\alpha_i]=e,
\end{equation}
with\, \medskip
\(
\alpha_i:=\begin{cases}
\delta_{i'}^{\,\varepsilon_1}&\scriptstyle i'=i\,\,\text{if}\,\,\varepsilon_1=1;\,\,i'=n+1-i\,\,\text{if}\,\,\varepsilon_1=-1,\textstyle\,\,\,\,\,\,\,\,\text{when}\,\,\,\,\,\,1\le i\le n
\\
\delta_{i'}^{\,\varepsilon_2}&\scriptstyle i'=i\,\,\text{if}\,\,\varepsilon_2=1;\,\,i'=2n+1-i\,\,\text{if}\,\,\varepsilon_2=-1,\textstyle\,\,\,\,\,\,\text{when}\,\,\,\,\,\,n+1\le i\le 2n
\\
\delta_{i'}^{\,\varepsilon_3}&\scriptstyle i'=i\,\,\text{if}\,\,\varepsilon_3=1;\,\,i'=3n+1-i\,\,\text{if}\,\,\varepsilon_3=-1,\textstyle\,\,\,\,\,\,\text{when}\,\,\,\,\,\,2n+1\le i\le 3n.
\end{cases}
\)

At this point, it is useful to note that, for any $\gamma\in L_{J_j}G$ and $\delta\in L_{J_k}G$, ($J_j,J_k\in \mathcal I$), the following equation holds:
\begin{equation}\label{eq: conj loops}
[\delta][\gamma][\delta]^{-1}=[\delta\gamma\delta^{-1}].
\end{equation}
If $k=j\pm 1$, this is true because $J_j\cup J_k\in\mathcal I$.
If $k\not=j\pm 1$, then $\gamma$ and $\delta$ have disjoint supports, $[\delta][\gamma][\delta]^{-1}[\gamma]^{-1}\in N$, and both sides of \eqref{eq: conj loops} are equal to $[\gamma]$.
Conjugation by a loop does not increase supports.
So we can iterate equation \eqref{eq: conj loops} to get:
\begin{equation}\label{eq: conj loops+}
[\delta_1]\ldots[\delta_s]\,[\gamma]\,[\delta_s]^{-1}\ldots[\delta_1]^{-1}=[\delta_1\ldots\delta_s\,\gamma\,\delta_s^{-1}\ldots\delta_1^{-1}].
\end{equation}

Let $\sigma\in\mathfrak S_{3n}$ be a permutation such that
$\alpha_{\sigma(3k-2)},\alpha_{\sigma(3k-1)},\alpha_{\sigma(3k)}\in L_{J_k}G $ for every $k\in\{1,\ldots,n\}$.
By Lemma \ref{lem: perm free gp},
there exist words $w_i\in (\mathrm{colim}_{\mathcal I}L_IG)/N$ in the $[\alpha_j]$'s so that
\[
\prod_{i=1}^{3n}\,[\alpha_i]=\prod_{i=1}^{3n}w_i[\alpha_{\sigma(i)}]w_i^{-1}.
\]
By \eqref{eq: conj loops+}, we have $w_i[\alpha_{\sigma(i)}]w_i^{-1}=[w_i\alpha_{\sigma(i)} w_i^{-1}]$
where, in the right hand side, we have identified $w_i$ with its image in $LG$.
Let $\beta_i:=\alpha_{\sigma(i)}$.
Recall that our goal is to show that equation \eqref{eq: our goal -- LG} holds.
So far, we have shown that
\[
\begin{split}
\prod_{i=1}^{3n}\,[\alpha_i]&=\prod_{i=1}^{3n}w_i[\beta_i]w_i^{-1}\\
&=
\prod_{k=1}^n
\big[ w_{3k-2}\beta_{3k-2} w_{3k-2}^{-1}\big]
\big[ w_{3k-1}\beta_{3k-1} w_{3k-1}^{-1}\big]
\big[ w_{3k}\beta_{3k} w_{3k}^{-1}\big]\\
&=
\prod_{k=1}^n
\big[ w_{3k-2}\beta_{3k-2} w_{3k-2}^{-1}
 w_{3k-1}\beta_{3k-1} w_{3k-1}^{-1}
 w_{3k}\beta_{3k} w_{3k}^{-1}\big].
\end{split}
\]
Letting $\chi_k:= w_{3k-2}\beta_{3k-2} w_{3k-2}^{-1}
 w_{3k-1}\beta_{3k-1} w_{3k-1}^{-1}
 w_{3k}\beta_{3k} w_{3k}^{-1}$, we rewrite this as:
\[
\prod_{i=1}^{3n}\,[\alpha_i]=\prod_{k=1}^n[\chi_k].
\]
By construction, $\mathrm{supp}(\chi_k)\subset J_k$.
Since $\chi_1 \chi_2 \ldots \chi_{n}=e$ in $LG$, 
and since $J_k\cap J_{k+2}=\emptyset$,
the support of each $\chi_k$ is contained in $(J_k\cap J_{k-1})\cup (J_k\cap J_{k+1})$.
We can thus write $\chi_k$ as $\chi_k=\chi^-_k\chi^+_k$ with $\mathrm{supp}(\chi^-_k)\subset J_k\cap J_{k-1}$ and  $\mathrm{supp}(\chi^+_k)\subset J_k\cap J_{k+1}$.
Finally, since $\chi_1\, \chi_2 \ldots \chi_{n}=e$, we have $\chi^+_k\chi^-_{k+1}=e$.
It follows that
\vspace{.2cm}
\[
\begin{split}
\tikzmath{
\node[inner ysep=-8] {$\displaystyle\prod_{i=1}^{3n}\,[\alpha_i]=\prod_{k=1}^n[\chi_k]$};}
&=[\chi^-_1\chi^+_1][\chi^-_2\chi^+_2]\ldots[\chi^-_n\chi^+_n]\\
&=[\chi^-_1][\chi^+_1][\chi^-_2][\chi^+_2]\ldots[\chi^-_n][\chi^+_n]\\
&=[\chi^-_1][\chi^+_1\chi^-_2][\chi^+_2\chi^-_3]\ldots[\chi^+_{n-1}\chi^-_n][\chi^+_n]\\
&=[\chi^-_1][\chi^+_n]=e.\qedhere
\end{split}
\]
\end{proof}

\begin{lemma}\label{lem: perm free gp}
Let $F_n=\langle a_1,\ldots,a_n\rangle$ be the free group on $n$ letters.
Then for any permutation $\sigma\in\mathfrak S_n$, there exist words $w_i\in F_n$ so that 
\[
\prod_{i=1}^{n}\,a_i=\prod_{i=1}^nw_i a_{\sigma(i)}w_i^{-1}.
\]
Moreover, the $w_i$ may be chosen so that each $a_i$ appears at most once in each~$w_i$.
\end{lemma}
\begin{proof}
Letting $w_1:=a_1\ldots a_{\sigma(1)-1}$, we have
\(
a_1\ldots a_n=
(w_1a_{\sigma(1)} w_i^{-1})\,
a_1\ldots \widehat{a_{\sigma(1)}}\ldots a_n
\).
Now use induction on $n$ to rewrite $a_1\ldots \widehat{a_{\sigma(1)}}\ldots a_n$ as
$\prod_{i=2}^nw_i a_{\sigma(i)}w_i^{-1}$.
\end{proof}

\subsubsection{The based loop group}

Fix a base point $p\in S$, and let $\Omega G\subset LG$ be the subgroup consisting of loops that map $p$ to the neutral element of $G$,
and all of whose derivatives vanish at that point.
We call $\Omega G$ the \emph{based loop group} of $G$.
Let $\widetilde{\Omega G}$ be the central extension of $\Omega G$ induced by the basic central extension~\eqref{eq: SES LG} of $LG$.

The arguments of the previous section can be adapted without difficulty to prove the following variants:
\[
\Omega G\,=\,\underset{I\subset S^1}{\mathrm{colim}}\,\, (L_IG\cap \Omega G),
\qquad\,\,\,\,\,
\widetilde{\Omega G}\,=\,\underset{I\subset S^1}{\mathrm{colim}}\,\, (\widetilde{L_IG}\cap \widetilde{\Omega G}).
\]
(The proofs are identical to those in the previous section:
replace every occurrence of $LG$ by $\Omega G$, and every occurrence of $L_IG$ by $L_IG\cap \Omega G$.)

It is also possible to express $\Omega G$ and $\widetilde{\Omega G}$ as colimits over
the poset of subintervals whose interior does not contain $p$, provided one works in the category of Hausdorff topological groups as opposed to the category of topological groups.

\begin{definition}
Given a diagram $\{G_i\}_{i\in\mathcal I}$ of Hausdorff topological groups, let us write $\mathrm{colim^{\scriptscriptstyle H}_{\mathcal I}}\,G_i$
for the colimit in Hausdorff topological groups.
Equivalently, this is the maximal Hausdorff quotient of $\mathrm{colim}_{\mathcal I}\,G_i$.
\end{definition}

Given an interval $I$, we write $\mathring I$ for its interior.
The next result does not seem to hold when the colimit is taken in the category of topological groups:

\begin{proposition}\label{prop: colim^H for LG}
The natural maps
\begin{equation}\label{eq: based loop gps as colimit}
\underset{I\subset S^1,\,p\;\!\not\in\;\!\mathring I\,\,\,}{\mathrm{colim^{\scriptscriptstyle H}}} L_IG\,\,\to\,\,\Omega G
\quad\qquad\text{and}\qquad\quad
\underset{I\subset S^1,\,p\;\!\not\in\;\!\mathring I\,\,\,}{\mathrm{colim^{\scriptscriptstyle H}}} \widetilde{L_IG}\,\,\to\,\,\widetilde{\Omega G}
\end{equation}
are isomorphisms of topological groups.
\end{proposition}

\begin{proof}
Let $\mathcal I$ be the poset of subintervals of $S$ whose interior does not contain $p$, and
let $N\triangleleft\, {\mathrm{colim}}_{\mathcal I}\,L_IG$ be the normal subgroup generated by commutators of loops whose supports have disjoint interiors.
The proof of Theorem~\ref{thm2: LG=colimit} applies verbatim (using a cover $\{J_j\}_{j=1...n}$ for which the $J_j\cap J_{j+1}$ are intervals for $0<j<n$, $J_0\cap J_n=\{p\}$, and all other intersections empty)
and shows that the map $(\mathrm{colim}_{\mathcal I}\, L_IG)\big/N\,\to\,\Omega G$ is an isomorphism of topological groups.
Since $\Omega G$ is Hausdorff, the natural map
\begin{equation*}
\Big(\underset{I\subset S^1,\,p\;\!\not\in\;\!\mathring I\,\,\,}{\mathrm{colim^{\scriptscriptstyle H}}}\, L_IG\Big)\big/N^{\scriptscriptstyle \mathrm H}\,\to\,\Omega G
\end{equation*}
is therefore an isomorphism, where $N^{\scriptscriptstyle \mathrm H}$ denotes the image of $N$ in $\mathrm{colim}^{\scriptscriptstyle \mathrm H}_{\mathcal I}L_IG$.

We wish to show that $N^{\scriptscriptstyle \mathrm H}$ is trivial.
Let $\gamma$ and $\delta$ be two loops whose supports have disjoint interiors. 
Write $S\setminus \{p\}$ as an increasing union of closed intervals $I_i\subset S$, and write
\[
\gamma=\lim \gamma_i\qquad \delta=\lim \delta_i,
\]
with $\mathrm{supp}(\gamma_i)\subset I_i$, $\mathrm{supp}(\delta_i)\subset I_i$, and $\mathrm{supp}(\gamma_i)\cap\mathrm{supp}(\delta_i)=\emptyset$.
The commutator $[\gamma_i,\delta_i]$ is trivial in $L_{I_i}G$, and therefore in ${\mathrm{colim}}_{\mathcal I}\,L_IG$.
By uniqueness of limits (this is where we use Hausdorffness\footnote{In the absence of the Hausdorffness condition, we could only deduce $[\gamma,\delta] \in\lim [\gamma_i,\delta_i]= \text{closure}(\{e\})$.}), it follows that $[\gamma,\delta] = [\lim \gamma_i,\lim \delta_i] = \lim [\gamma_i,\delta_i] = \lim e = e$ in $\mathrm{colim}^{\scriptscriptstyle \mathrm H}_{\mathcal I}L_IG$.
This show that $N^{\scriptscriptstyle \mathrm H}$ is the trivial group, and that the first map in \eqref{eq: based loop gps as colimit} is an isomorphism.

Proposition~\ref{prop: central ex of colim} is stated in the category of topological groups, but it also holds in the category of Hausdorff topological groups (with identical proof: just replace every occurrence of $\mathrm{colim}$ by $\mathrm{colim}^{\scriptscriptstyle \mathrm H}$).
The first isomorphism in \eqref{eq: based loop gps as colimit} therefore implies the second one.
\end{proof}

Similarly, letting $\widetilde{\Omega G}^k\!$ be the pullback of $\widetilde{LG}^k\!$ along the inclusion $\Omega G\to LG$, the natural map
\begin{equation}\label{eq: based loop gps as colimit -- BIS}
\underset{I\subset S^1,\,p\;\!\not\in\;\!\mathring I\,\,\,}{\mathrm{colim^{\scriptscriptstyle H}}} \widetilde{L_IG}^k\,\,\to\,\,\widetilde{\Omega G}^k
\end{equation}
is an isomorphism of topological groups.
The proof that this map is an isomorphism is identical to that of the second isomorphism in \eqref{eq: based loop gps as colimit}.

\subsection{Diffeomorphism groups}\label{sec:Diffeomorphism groups}

The material in this section is largely parallel to the one in the previous section, with one notable difference.
Whereas conjugating by a loop never increases the support, conjugating by a diffeomorphism does typically increase supports.
This introduces a number of small subtleties.

Recall that we write $S^1:=\{z\in\C:|z|=1\}$ for the standard circle, and $S$ for a manifold diffeomorphic to $S^1$.
All our circles are assumed oriented.

\subsubsection{$\Diff(S^1)$ and its universal cover}

Given an interval $I$, we write $\Diff_0(I)\subset \Diff(I)$ for the group of diffeomorphisms  of $I$ that are tangent up to infinite order to the identity map at the two boundary points.
If $I$ is a subinterval of a circle $S$, then this group can be equivalently described as the subgroup $\Diff_0(I)\subset \Diff_+(S)$ of diffeomorphisms with support in $I$.

\begin{theorem}\label{thm: Diff=colimit}
Let $S$ be a circle, and
let $\tilde{\Diff_+}(S)$ be the universal cover of the group of orientation preserving diffeomorphisms of $S$.
Then the natural map
\[
{\mathrm{colim}}_{I\subset S}\, \Diff_0(I)\,\to\,\tilde{\Diff_+}(S)
\]
is an isomorphism of topological groups.
\end{theorem}

The Lie algebra $\mathfrak X(S)$ of smooth vector fields on $S$ has a well known central extension by $i\R$, constructed as follows.
Upon identifying $S$ with $S^1$, it can be described as the central extension associated to the $2$-cocycle $(f\tfrac{\partial}{\partial z},g\tfrac{\partial}{\partial z})\mapsto 
\tfrac1{12}\int_{S^1}\tfrac{\partial f}{\partial z}(z)\,\tfrac{\partial^2 g}{\partial z^2}(z)\,\tfrac{dz}{2\pi i}$.
In terms of the topological basis $\ell_n:=-z^{n+1}\tfrac{\partial}{\partial z}$ of the complexified Lie algebra, this cocycle can also be described by the formula:\footnote{The cocycle \eqref{vircoc} is cohomologous to $(\ell_m,\ell_n) \mapsto \tfrac{1}{12}\cdot m^3 \delta_{m+n,0}$, but the former is usually preferred because it is $\mathit{PSL}(2,\R)$-invariant.}
\begin{equation}\label{vircoc}
(\ell_m,\ell_n) \mapsto \tfrac1{12}(m^3-m)\delta_{m+n,0}.
\end{equation}
The corresponding central extension of $\mathfrak X(S^1)$ is a universal central extension in the category of topological Lie algebras \cite{MR0245035}.
Since $\mathfrak X(S)$ and $\mathfrak X(S^1)$ are isomorphic as topological Lie algebras, the former also admits a universal central extension by~$i\R$ (universal central extensions are well defined up to unique isomorphism).

Finally, the universal central extension of $\mathfrak X(S)$ integrates to a central extension $0\to\R\to\Diff_+^\R(S)\to\Diff_+(S)\to 0$,
called the Virasoro-Bott group of $S$ \cite[Chapt II.2]{MR2456522}\cite{MR0488080}\cite{MR1674631} (where we have identified $i\R$ with $\R$ for notational convenience).
Let $\Diff_+^{\R\times \Z}(S)$ be the universal cover of the Virasoro-Bott group.

\begin{theorem}\label{thm: Diff=colimit -- BIS}
Let $S$ be a circle.
Given an interval $I\subset S$, let us denote by $\Diff_0^\R(I)$ the pullback of the central extension $\Diff_+^\R(S) \to \Diff_+(S)$
along the inclusion $\Diff_0(I)\to \Diff_+(S)$.
Then the natural map
\[
{\mathrm{colim}}_{I\subset S}\,\, \Diff_0^\R(I)\,\to\,\Diff_+^{\R\times \Z}(S)
\]
is an isomorphism of topological groups.
\end{theorem}

\begin{remark}\label{rem: Bott-Virasoro non-trivial?}
The central extension $\Diff_+^{\R}(S^1) \to \Diff_+(S^1)$ is non-trivial not only as an extension of Lie groups, but also as an extension of abstract groups.
To see this, one can argue as follows:

Let $\mathit{PSL}(2,\R)^{(n)}$ be the subgroup of $\Diff_+(S^1)$ corresponding to the subalgebra $\mathrm{Span}_{\R}\{i\ell_0,\ell_n-\ell_{-n},i\ell_n+i\ell_{-n}\}$ of $\mathfrak X(S^1)$.
That Lie algebra lifts to a Lie algebra $\mathrm{Span}_{\R}\{(i\ell_0,\tfrac{i}{12}{\cdot}\tfrac{n^2-1}{2}),(\ell_n-\ell_{-n},0),(i\ell_n+i\ell_{-n},0)\}$ in the central extension of $\mathfrak X(S^1)$.
The latter integrates to a subgroup of the Virasoro-Bott group isomorphic to $\mathit{PSL}(2,\R)^{(n)}$:
\begin{equation}\label{eq: lift of PSL2}
\tikzmath{
\node (A) at (0,0) {$\mathit{PSL}(2,\R)^{(n)}$};
\node (B) at (3,1.5) {$\Diff_+^{\R}(S^1)$};
\node (C) at (3,0) {$\Diff_+(S^1)$};
\draw [->, dashed] (A) --node[xshift=-3, above]{$\scriptstyle s_n$} (B);
\draw [->] (B) -- (C);
\draw [->] (A) -- (C);
}
\end{equation}
Moreover, since $\mathit{PSL}(2,\R)^{(n)}$ has trivial abelianization\footnote{
The commutators
$\big[(\begin{smallmatrix}\lambda & 0 \\ 0 & 1/\lambda\end{smallmatrix}),
(\begin{smallmatrix}1 & a \\ 0 & 1\end{smallmatrix})
\big]$
and
$\big[(\begin{smallmatrix}\lambda & 0 \\ 0 & 1/\lambda\end{smallmatrix}),
(\begin{smallmatrix}1 & 0 \\ b & 1\end{smallmatrix})
\big]$
generate a neighbourhood of the identity in $\mathit{PSL}(2,\R)$.
}, the lift $s_n$ is unique (without any continuity assumptions).

Assume by contradiction that there exists a section $s:\Diff_+(S)\to \Diff_+^{\R}(S)$ which is a group homomorphism, possibly discontinuous.
By uniqueness of the lift \eqref{eq: lift of PSL2}, we would then have $s|_{\mathit{PSL}(2,\R)^{(n)}}=s_n$,
from which it would that follow that
\[
s_n|_{\mathit{PSL}(2,\R)^{(n)}\cap\mathit{PSL}(2,\R)^{(m)}}\,=\,s_m|_{\mathit{PSL}(2,\R)^{(n)}\cap\mathit{PSL}(2,\R)^{(m)}}.
\]
But $\mathit{PSL}(2,\R)^{(n)}\cap\mathit{PSL}(2,\R)^{(m)}$ is the circle subgroup with Lie algebra $\mathrm{Span}_{\R}\{i\ell_0\}$, and one can easily check at the Lie algebra level that
$s_n|_{S^1}\not=s_m|_{S^1}$.

A similar argument shows that the central extension $\Diff_+^{\R\times \Z}(S) \to \tilde{\Diff_+}(S)$ remains non-trivial when viewed as an extension of abstract groups.
\end{remark}

The next remark provides an answer to a question by Vaughan Jones:

\begin{remark}
Theorem~\ref{thm: Diff=colimit -- BIS} 
can be used to show that the central extension $\Diff_+^{\R}(S)\to \Diff_+(S)$ remains non-trivial upon restricting it to a subgroup $\Diff_0(I)$.
We first note that, since $\Diff_0(I)$ is perfect \cite{MR1882783}, for any $J\subset I\subset S$, the inclusion map $\iota^\R_{\scriptscriptstyle JI}:\Diff_0^\R(J)\to \Diff_0^\R(I)$ is uniquely characterized by the fact that $\iota^\R_{\scriptscriptstyle JI}|_{\R}=\mathrm{id}_\R$, and that it covers the inclusion map $\iota_{\scriptscriptstyle JI}:\Diff_0(J)\to \Diff_0(I)$.

Suppose by contradiction that the central extension was trivial: $\Diff_0^\R(I)\cong \R\times \Diff_0(I)$.
Then we would have $\iota^\R_{\scriptscriptstyle JI}=\mathrm{id}_\R\times \iota_{\scriptscriptstyle JI}$,
from which it would follow that
\[
\begin{split}
\Diff_+^{\R\times \Z}(S)\,&\cong\,{\mathrm{colim}}_{I\subset S}\,\, \Diff_0^\R(I)
\\&\cong\,
{\mathrm{colim}}_{I\subset S}\,\, \big(\R\times \Diff_0(I)\big)
\\&\cong\,
\R\times {\mathrm{colim}}_{I\subset S}\,\, \Diff_0(I)\,\cong\,
\R\times \tilde{\Diff_+}(S),
\end{split}
\]
contradicting the non-triviality of the central extension $\Diff_+^{\R\times \Z}(S) \to \tilde{\Diff_+}(S)$.
\end{remark}

The main technical tool in our proofs of Theorems~\ref{thm: Diff=colimit} and~\ref{thm: Diff=colimit -- BIS} is a kind of partitions of unity for diffeomorphisms.
The result is very similar to \cite[Lem.\,3]{MR2078164}.
Let
\[
\Diff^{<d}(\R):=\,\{\varphi\in\Diff(\R):|\varphi(t)-t|<d\}
\]
be the set of diffeomorphisms of displacement smaller than $d$.
Similarly, let $\Diff^{<d}(S^1):=\,\{\varphi\in\Diff(S^1):|\varphi(t)-t|<d\}$.

\begin{lemma}\label{lem: part of un}
{\rm(i)}
There exist continuous maps $\!\tikzmath{\node[scale=.97]{$(\,\,\,)_-$};}\!,\!\tikzmath{\node[scale=.97]{$(\,\,\,)_+$};}\!\!:\Diff^{<d}(\R)\to \Diff^{<d}(\R)$ such that,
for every $\varphi \in \Diff^{<d}(\R)$, we have:
\[
\qquad\quad\varphi\,=\,\varphi_- \varphi_+,\,\quad \mathrm{supp}(\varphi_-)\subset(-\infty,d\,],\quad \mathrm{supp}(\varphi_+)\subset[-d,\infty).
\]
{\rm(ii)}
Let $I_-, I_+\subset S^1$ be two subintervals that cover the standard circle. Assume that each connected component of $I_-\cap I_+$ has length $2d$.
Then there exist continuous maps $\!\tikzmath{\node[scale=.97]{$(\,\,\,)_-$};}\!,\!\tikzmath{\node[scale=.97]{$(\,\,\,)_+$};}\!\!:\Diff^{<d}(S^1)\to \Diff^{<d}(S^1)$ such that, 
for every $\varphi \in \Diff^{<d}(S^1)$, we have:
\[
\qquad\quad\varphi\,=\,\varphi_- \varphi_+,\,\quad \mathrm{supp}(\varphi_-)\subset I_-,\quad \mathrm{supp}(\varphi_+)\subset I_+.
\]
\end{lemma}

\begin{proof}
We only prove the fist part of the lemma
(the second part is completely analogous).
Let $\sigma:\R\to [0,d]$ be a monotonic function such that $\sigma(t)=0$ on $(-\infty,-d]$, $\sigma(t)=d$ on $[d,\infty)$, and $\sigma'(t)<1$.
Given $\varphi \in \Diff^{<d}(\R)$, we let $\varphi_+$ be the unique solution of the functional equation
$\tfrac{\varphi_+(t)-t}{\varphi(t)-t}=\tfrac{\sigma(\varphi_+(t))}d$:
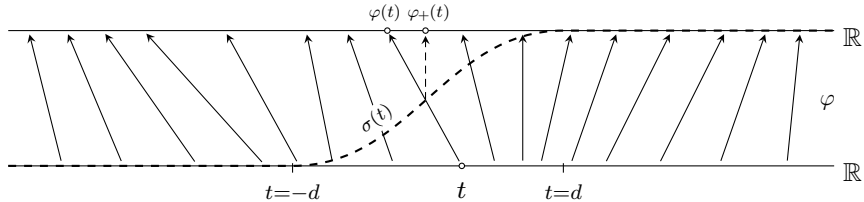
\[
\tikzmath[scale=1.8]{
\draw(-2.1,0) -- (4,0)node[right, scale=.9, yshift=-2.5]{$\R$}(-2.1,1) -- (4,1)node[right, scale=.9, yshift=-2.5]{$\R$};
\node[scale=.8] at (3.95,.47) {$\varphi$};
\foreach \P in {{(-1.7,0) -- (-1.95,1)},{(-1.25,0) -- (-1.67,1)},{(-.7,0) -- (-1.4,1)},{(-.2,0) -- (-1.1,1)},{(0.05,0) -- (-.5,1)},{(.3,0) -- (0.1,1)},{(.75,0) -- (.4,1)},{(1.25,0) -- (.7,1)},{(1.5,0) -- (1.25,1)},{(1.7,0) -- (1.7,1)},{(1.83,0) -- (2.06,1)},{(2.05,0) -- (2.4,1)},{(2.3,0) -- (2.8,1)},{(2.7,0) -- (3.2,1)},{(3.15,0) -- (3.5,1)},{(3.65,0) -- (3.75,1)}}
{\draw[-stealth, shorten > = 1.5, shorten < = 2] \P;}
\draw (0,.02) -- (0,-.05)node[below, yshift=-.5]{$\scriptstyle t=-d$};
\draw (2,.02) -- (2,-.05)node[below, yshift=-.5]{$\scriptstyle t=d$};
\draw[fill=white] (1.25,0) circle (.02); 
\node[below, scale=1.2, yshift=.7] at (1.25,-.05) {$\scriptstyle t$};
\draw[fill=white] (.7,1) circle (.02); 
\draw[thick, dashed] (-2.1,0) -- (0,0) to[out=0,in=-180] (2,1) -- (4,1);
\draw[densely dashed, -stealth, shorten > =2] (.982,.487) -- (.982,1);
\node[above, scale=.8, yshift=.7, xshift=-1.2] at (.7,1) {$\scriptstyle \varphi(t)$};
\draw[fill=white] (.982,1) circle (.02) node[above, xshift=1.2, scale=.8, yshift=.7] {$\scriptstyle \varphi_{\hspace{-.2mm}+}\hspace{-.2mm}(t)$};
\node[scale=.9, rotate=35, fill=white, inner ysep=1.5] at (.62,.35) {$\scriptstyle \sigma(t)$};
}
\]
It has support in $[-d,\infty)$, and satisfies $\varphi_+^{-1}(t)=\varphi^{-1}(t)$ for every $t\ge d$.
The diffeomorphism $\varphi_-:=\varphi\,\varphi_+^{-1}$ has displacement smaller than $d$, and support in $(-\infty,d\,]$.
\end{proof}

Given a subinterval $I$ of $\R$ or $S^1$, let $\Diff^{<d}_0(I) := \Diff_0(I)\cap \Diff^{<d}(I)$.

\begin{corollary}\label{cor: part of un on S^1}
Let $\{I_i\}_{i\,=\,1\ldots n}$ be a cover of $S^1$ such that each intersection $I_i\cap I_{i+1}$ (cyclic numbering) has length $2d$ for some $d$, and the other intersections are empty:
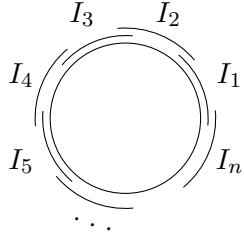
\[
\tikzmath{
\draw circle (1);
\draw (-5:1.1) arc (-5:50:1.1);
\draw (40:1.2) arc (40:95:1.2);
\draw (90-5:1.1) arc (90-5:90+50:1.1);
\draw (90+40:1.2) arc (90+40:90+95:1.2);
\draw (-5:-1.1) arc (-5:50:-1.1);
\draw (40:-1.2) arc (40:95:-1.2);
\draw (90+40:-1.2) arc (90+40:90+95:-1.2);
\foreach \r in {1,2,3,4,5}
{\node at (\r*45-22.5:1.5) {$I_\r$};}
\node at (8*45-22.5:1.5) {$I_n$};
\node at (-90-8.5:1.5) {$\cdot$};
\node at (-90-8.5-8.5:1.5) {$\cdot$};
\node at (-90-8.5-8.5-8.5:1.5) {$\cdot$};
}
\]
Then any element $\varphi\in\Diff^{<d}(S^1)$ can be factored as
$\varphi=\varphi_1\ldots  \varphi_n$, with $\varphi_j\in \Diff_0^{<d}(I_j)$.
Moreover, the factorisation can be chosen to depend continuously on~$\varphi$.
\end{corollary}

\begin{proof}
Apply Lemma \ref{lem: part of un}(ii) to write $\varphi=\varphi_1\varphi_+$ with $\varphi_1\in\Diff^{<d}_0(I_1)$ and $\varphi_+\in\Diff^{<d}_0(I_2\cup\ldots\cup I_n)$.
Then use Lemma \ref{lem: part of un}(i) to write $\varphi_+=\varphi_2\ldots \varphi_n$ with $\varphi_i\in\Diff^{<d}_0(I_i)$, for $2\le i\le n$.
\end{proof}

\begin{lemma}\label{lem: diff is genr'ted}
{\rm(i)}
Let $I$ be an interval and let $\mathcal I=\{I_i\}$ be a finite collection of subintervals whose interiors cover that of $I$.
Then the subgroups $\Diff_0(I_i)$ generate $\Diff_0(I)$.\\
{\rm(ii)}
Let $S$ be a circle and let $\mathcal I=\{I_i\}$ be a collection of subintervals whose interiors cover $S$.
Then the subgroups $\Diff_0(I_i)$ generate $\tilde \Diff_+(S)$.
\end{lemma}

\begin{proof} 
We only prove the second statement (the first one is completely analogous).
Assume without loss of generality that $S=S^1$.
Let $\{J_j\}_{j\,=\,1\ldots n}$ be a refinement of our cover such that each $J_j\cap J_{j+1}$ has length $2d$, for some constant $d>0$, and the other intersections are empty (cyclic numbering).
Write $\varphi$ as a product $\varphi_1\ldots \varphi_m$ of diffeomorphisms of displacement smaller than $d$.
Now apply Corollary~\ref{cor: part of un on S^1} to each $\varphi_i$ to rewrite it as a product $\varphi_i=\varphi_{i,1}\ldots \varphi_{i,n}$, with $\varphi_{i,j}\in\Diff_0(J_j)$.
\end{proof}

Let $\mathcal I$ be a collection of intervals in $S^1$ whose interiors form a cover, and that is closed under taking subintervals.
As in \eqref{eq: LG two ways}, for any $I_1$, $I_2 \in \mathcal I$, the diagram 
\begin{equation}\label{eq: diff two ways}
\tikzmath{
\node (A) at (2,0) {$\Diff_0(I_1\cap I_2)$};
\node (B) at (5.1,.8) {$\Diff_0(I_1)$};
\node (C) at (5.1,-.8) {$\Diff_0(I_2)$};
\node (D) at (8.2,0) {${\mathrm{colim}}_{\mathcal I}\,\Diff_0(I)$};
\draw [->] (A) -- (B);
\draw [->] (A) -- (C);
\draw [<-] (D.north west) + (.2,-.03) --node[above, pos=.45]{$\scriptstyle \iota_1$} (B);
\draw [<-] (D.south west) + (.2,.03) --node[below, pos=.45]{$\scriptstyle \iota_2$} (C);
}
\end{equation}
commutes.
Given a diffeomorphism $\varphi\in\Diff_+(S^1)$ with support in some interval $I\in\mathcal I$, we also write $\varphi$ for its image in ${\mathrm{colim}}_{\mathcal I}\,\Diff_0(I)$.
This element is well-defined by the commutativity of \eqref{eq: diff two ways}.

The following result is a strengthening of Theorem~\ref{thm: Diff=colimit}:

\begin{theorem}\label{thm2: Diff=colimit}
Let $\mathcal I$ be a collection of intervals in $S^1$ whose interiors form a cover, and that is closed under taking subintervals.
Let $N\triangleleft\, {\mathrm{colim}}_{\mathcal I}\,\Diff_0(I)$
be the normal subgroup generated by commutators of diffeomorphisms with disjoint supports.
Then the natural map
\begin{equation}\label{eq: diff is a colimit+++}
\big({\mathrm{colim}}_{\mathcal I}\, \Diff_0(I)\big)/N\,\to\,\tilde{\Diff_+}(S^1)
\end{equation}
to the universal cover of $\Diff_+(S^1)$ is an isomorphism of topological groups.
\end{theorem}

Before embarking in the proof, let us show how Theorems~\ref{thm: Diff=colimit} and~\ref{thm: Diff=colimit -- BIS} follow from the above result.

\begin{proof}[Proof of Theorem~\ref{thm: Diff=colimit}]
Without loss of generality, we take $S=S^1$. 
Let $\mathcal I$ be the poset of all subintervals of $S^1$.
By Theorem \ref{thm2: Diff=colimit}, the map $({\mathrm{colim}}_{I\subset S^1}\, \Diff_0(I))/N\,\to\,\tilde{\Diff_+}(S^1)$ is an isomorphism.
So it suffices to show that $N$ is trivial.
Given diffeomorphisms with disjoint support $\varphi,\psi\in\Diff_+(S^1)$, there exists an interval $I\subset S^1$ such that both $\varphi$ and $\psi$ are in $\Diff_0(I)$.
The commutator of $\varphi$ and $\psi$ is trivial in $\Diff_0(I)$. It is therefore also trivial in the colimit.
\end{proof}

\begin{proof}[Proof of Theorem~\ref{thm: Diff=colimit -- BIS}]
The maps 
$\!\tikzmath{\node[scale=.97]{$(\,\,\,)_-$};}\!,\!\tikzmath{\node[scale=.97]{$(\,\,\,)_+$};}\!\!:\Diff^{<d}(S^1)\to \Diff(S^1)$ in Lemma~\ref{lem: part of un}(ii)
show that the colimit which appears in Theorem~\ref{thm: Diff=colimit} satisfies the assumption of Proposition~\ref{prop: central ex of colim}.
Theorem~\ref{thm: Diff=colimit -- BIS} therefore follows from Theorem~\ref{thm: Diff=colimit}.
\end{proof}

\begin{proof}[Proof of Theorem \ref{thm2: Diff=colimit}]
Given a diffeomorphism $\varphi\in\Diff_+(S^1)$ with support in some interval $I\in\mathcal I$, we write $[\varphi]$ for its image in the group $({\mathrm{colim}}_{\mathcal I}\,\Diff_0(I))/N$.

Let $\{J_j\}_{j\,=\,1\ldots n}$ be a cover of $S^1$ such that each intersection $J_j\cap J_{j+1}$ (cyclic numbering) has length $2d$ for some $d$, and the other intersections are empty.
The $J_j$ are chosen so that each $J_{j-1}\cup J_{j}\cup J_{j+1}$ is in $\mathcal I$ (cyclic numbering)
and the distance between $J_j$ and $J_{j+2}$ is greater than $6d$.
By Corollary~\ref{cor: part of un on S^1}, any element $\varphi\in\Diff^{<d}(S^1)$ can be factored as
$\varphi=\varphi_1\ldots  \varphi_n$, with $\varphi_j\in \Diff_0^{<d}(J_j)$.
Moreover, that factorisation may be chosen to depend continuously on~$\varphi$.
After identifying $\Diff^{<d}(S^1)$ with an open subset of $\tilde{\Diff_+}(S^1)$,
this provides a local section of the map in~\eqref{eq: diff is a colimit+++}:
\begin{equation*}
\tikzmath{
\node[below] (A) at (0,0) {$(\mathrm{colim}_{\mathcal I}\,\Diff_0(I))/N$}; 
\node[below] (B) at (4.2,.075) {$\tilde{\Diff_+}(S^1)$.};
\draw[->] (A.east) -- (B.west|-A.east);
\node[inner sep=3] (C) at (2.5,-1.5) {$\Diff^{<d}(S^1)\,\,$};
\draw[<-right hook, shorten >=-.8] (B.south)+(-.69,0) -- (C);
\draw[dashed, ->, shorten <=2.5] (C) -- (1,-.75);
\node at (0,-1) {$\scriptstyle [\varphi_1]\ldots [\varphi_n]$};
\node at                    (.83,-1.48) {$\scriptstyle \varphi$};
\node[rotate=154] at (.5,-1.3) {$\scriptstyle \mapsto$};
}
\end{equation*}
The map \eqref{eq: diff is a colimit+++} is surjective by Lemma~\ref{lem: diff is genr'ted}, and admits continuous local sections.

It remains to prove injectivity.
Let $g\in (\mathrm{colim}_{\mathcal I}\Diff_0(I))/N$ be in the kernel of the map to $\tilde{\Diff_+}(S^1)$.
By Lemma~\ref{lem: diff is genr'ted}, we may rewrite it as a product
\(
[\varphi_1][\varphi_2]\ldots[\varphi_N]
\),
where each $\varphi_i$ is in some $\Diff_0^{<d}(J_j)$.
By assumption, the relation
\begin{equation}\label{eq: rel1 -- PRE}
\varphi_1 \varphi_2 \ldots \varphi_N=e
\end{equation}
holds in $\tilde{\Diff_+}(S^1)$. Our goal is to show that the relation 
\begin{equation}\label{eq: rel1}
[\varphi_1][\varphi_2] \ldots [\varphi_N]=e
\end{equation}
holds in $(\mathrm{colim}_{\mathcal I}\,\Diff_0(I))/N$.

Any $x\in\Diff^{<d}(S^1)$ can be factored as $x=x_1\ldots x_n$ with $x_j\in \Diff_0^{<d}(J_j)$,
so the set $\mathcal U:=\Diff_0^{<d}(J_1)\ldots \Diff_0^{<d}(J_n) =\{x_i\ldots x_n:x_j\in \Diff_0^{<d}(J_j)\}$ is a neighbourhood of $e\in \tilde{\Diff_+}(S^1)$.
The set $\mathcal U$ is visibly path-connected.
By Lemma~\ref{lem: vK diagr EZ}, 
equation \eqref{eq: rel1 -- PRE} is therefore a formal consequence of certain relations of length $3$ between the elements of~$\mathcal U$:
\begin{equation}\label{eq: rel2 -- PRE}
(\psi_1\, \psi_2 \ldots \psi_{n})^{\varepsilon_1}(\psi_{n+1} \ldots \psi_{2n})^{\varepsilon_2}(\psi_{2n+1} \ldots \psi_{3n})^{\varepsilon_3}=e,
\end{equation}
$\psi_i$, $\psi_{n+i}$, $\psi_{2n+i}\in \Diff_0^{<d}(J_i)$.
In order to prove that \eqref{eq: rel1} holds,
it is therefore enough to show that the relations
\begin{equation}\label{eq: rel2}
\big([\psi_1][\psi_2] \ldots [\psi_{n}]\big)^{\varepsilon_1}\big([\psi_{n+1}] \ldots [\psi_{2n}]\big)^{\varepsilon_2}\big([\psi_{2n+1}] \ldots [\psi_{3n}]\big)^{\varepsilon_3}=e
\end{equation}
hold in $(\mathrm{colim}_{\mathcal I}\Diff_0(I))/N$.
Using that $[\psi_i]^{-1}=[\psi_i^{-1}]$,
we rewrite \eqref{eq: rel2} as
\begin{equation}\label{eq: our goal}
\prod_{i=1}^{3n}\,[\alpha_i]=e,
\end{equation}
with\, \medskip
\(
\alpha_i:=\begin{cases}
\psi_{i'}^{\,\varepsilon_1}&\scriptstyle i':=i\,\,\text{if}\,\,\varepsilon_1=1;\,i':=n+1-i\,\,\text{if}\,\,\varepsilon_1=-1,\textstyle\,\,\,\,\,\,\,\,\text{when}\,\,\,\,\,\,1\le i\le n
\\
\psi_{i'}^{\,\varepsilon_2}&\scriptstyle i':=i\,\,\text{if}\,\,\varepsilon_2=1;\,i':=2n+1-i\,\,\text{if}\,\,\varepsilon_2=-1,\textstyle\,\,\,\,\,\,\text{when}\,\,\,\,\,\,n+1\le i\le 2n
\\
\psi_{i'}^{\,\varepsilon_3}&\scriptstyle i':=i\,\,\text{if}\,\,\varepsilon_3=1;\,i':=3n+1-i\,\,\text{if}\,\,\varepsilon_3=-1,\textstyle\,\,\,\,\,\,\text{when}\,\,\,\,\,\,2n+1\le i\le 3n.
\end{cases}
\)

As in \eqref{eq: conj loops}, for any $\varphi\in\Diff_0(J_j)$ and $\psi\in\Diff_0(J_k)$, the following equation holds:
\begin{equation}\label{eq: conj diffeos}
[\psi][\varphi][\psi]^{-1}=[\psi\varphi\psi^{-1}].
\end{equation}
We would like to replace $[\psi]$ in \eqref{eq: conj diffeos} by an arbitrary word $[\psi_1]\ldots[\psi_s]$:
\begin{equation}\label{eq: conj diffeos+}
[\psi_1]\ldots[\psi_s]\,[\varphi]\,[\psi_s]^{-1}\ldots[\psi_1]^{-1}=[\psi_1\ldots\psi_s\,\varphi\,\psi_s^{-1}\ldots\psi_1^{-1}].
\end{equation}
However, for general $\psi_i\in \bigcup_k\Diff_0(J_k)$ it is not clear that \eqref{eq: conj diffeos+} should hold, because each time one conjugates $\varphi$ by a diffeomorphism, its support grows.
If we insist, however, that the $\psi_i$ have small displacement, so as to control the supports of $\psi_r\psi_{r+1}\ldots\psi_s\,\varphi\,\psi_s^{-1}\ldots\psi_{r+1}^{-1}\psi_{r}^{-1}$,
then equation \eqref{eq: conj diffeos+} will hold.
The precise version of \eqref{eq: conj diffeos+} that we will need is he following:
{\it Let $\psi_i\in \bigcup_k\Diff_0^{<d}(J_k)$, and let $\varphi\in\Diff_0(J_j)$.
If there are at most three $\psi_i$'s whose support is in $J_{j-1}$ and not in some other $J_k$, and at most three $\psi_i$'s whose support $J_{j+1}$ and not in some other $J_k$, then equation \eqref{eq: conj diffeos+} holds.}
The proof is an iteration of the argument used for \eqref{eq: conj diffeos}, while keeping track of the size of the supports.
(This only uses the fact that the distance between $J_j$ and $J_{j+2}$ is greater than $3d$.
Later, we will use that it is greater than $6d$.)

Let $\sigma\in\mathfrak S_{3n}$ be a permutation such that $\alpha_{\sigma(3k-2)},\alpha_{\sigma(3k-1)},\alpha_{\sigma(3k)}\in \Diff_0^{<d}(J_k) $.
By Lemma \ref{lem: perm free gp}, there exist words $w_i$ in the $[\alpha_j]$'s so that $\prod^{3n}\,[\alpha_i]=\prod^{3n}w_i[\alpha_{\sigma(i)}]w_i^{-1}$.
Moreover, these words can be chosen so that each $[\alpha_i]$ appears at most once in each $w_i$.
By \eqref{eq: conj diffeos+}, we then have:
\[
w_i[\alpha_{\sigma(i)}]w_i^{-1}=[w_i\alpha_{\sigma(i)}w_i^{-1}].
\]
Recall that our goal is to show that \eqref{eq: our goal} holds.
Let $\beta_i:=\alpha_{\sigma(i)}$.
So far, we have:
\[
\begin{split}
\prod_{i=1}^{3n}\,[\alpha_i]
=\prod_{i=1}^{3n}w_i[\beta_i]w_i^{-1}
&=
\prod_{k=1}^n
\big[ w_{3k-2}\beta_{3k-2} w_{3k-2}^{-1}\big]
\big[ w_{3k-1}\beta_{3k-1} w_{3k-1}^{-1}\big]
\big[ w_{3k}\beta_{3k} w_{3k}^{-1}\big]\\
&=
\prod_{k=1}^n
\big[ w_{3k-2}\beta_{3k-2} w_{3k-2}^{-1}
 w_{3k-1}\beta_{3k-1} w_{3k-1}^{-1}
 w_{3k}\beta_{3k} w_{3k}^{-1}\big].
\end{split}
\]
Let $\chi_k:= w_{3k-2}\beta_{3k-2} w_{3k-2}^{-1} w_{3k-1}\beta_{3k-1} w_{3k-1}^{-1} w_{3k}\beta_{3k} w_{3k}^{-1}$ so that $\prod^{3n}\,[\alpha_i]=\prod^n[\chi_k]$.
By construction, $\mathrm{supp}(\chi_k)\subset J_k^+$, where $J_k^+$ is obtained from $J_k$ by by enlarging it by $3d$ on each side.
The crucial property of those slightly larger intervals is that $J_k^+$ does not overlap with $J_{k+2}^+$.
The relation $\chi_1 \chi_2 \ldots \chi_{n}=e$ holds in $\tilde{\Diff_+}(S^1)$ so the support of each $\chi_k$ is contained in $(J_k^+\cap J_{k-1}^+)\cup (J_k^+\cap J_{k+1}^+)$.
We can thus write $\chi_k$ as $\chi_k=\chi^-_k\chi^+_k$, with $\mathrm{supp}(\chi^-_k)\subset J_k^+\cap J_{k-1}^+$ and  $\mathrm{supp}(\chi^+_k)\subset J_k^+\cap J_{k+1}^+$.
Since $\chi_1\, \chi_2 \ldots \chi_{n}=e$, we must have $\chi^+_k\chi^-_{k+1}=e$.
Finally, as in the proof of Theorem \ref{thm2: LG=colimit},\vspace{.2cm}
\[
\begin{split}
\tikzmath{
\node[inner ysep=-8] {$\displaystyle\prod_{i=1}^{3n}\,[\alpha_i]=\prod_{k=1}^n[\chi_k]$};}
&=[\chi^-_1\chi^+_1][\chi^-_2\chi^+_2]\ldots[\chi^-_n\chi^+_n]\\
&=[\chi^-_1][\chi^+_1\chi^-_2][\chi^+_2\chi^-_3]\ldots[\chi^+_{n-1}\chi^-_n][\chi^+_n]=e.
\end{split}
\]
\end{proof}

\subsubsection{The based diffeomorphism group}

Choose a base point $p\in S$, and let $\Diff_*(S)\subset \Diff(S)$ be the subgroup of diffeomorphisms that fix $p$ and that are tangent to $\mathrm{id}_S$ up to infinite order at that point. 
We call this group the \emph{based diffeomorphism group} of $S$.
Let $\Diff_*^\R(S)$ be the restriction of the central extension by $\R$ to $\Diff_*(S)$.
The arguments of the previous section can be adapted without difficulty to prove the following variants:
\begin{equation*}
\Diff_*(S)=\underset{\substack{I\subset S}}{\mathrm{colim}}\, \big(\Diff_0(I)\cap \Diff_*(S)\big),
\quad\,\,\,
\Diff_*^\R(S)=\underset{\substack{I\subset S}}{\mathrm{colim}}\, \big(\Diff_0^\R(I)\cap \Diff_*^\R(S)\big).
\end{equation*}
The proofs are identical to those in the previous section:
replace every occurrence of $\Diff(S^1)$ by $\Diff_*(S^1)$, and every occurrence of $\Diff_0(I)$ by $\Diff_0(I)\cap \Diff_*(S^1)$.

It is also possible to express the groups $\Diff_*(S)$ and $\Diff_*^\R(S)$ as colimits over
the poset of subintervals whose interior does not contain $p$, provided one works in the category of Hausdorff topological groups:
\begin{proposition}\label{prop: based Diff colimit}
The natural maps
\begin{equation}\label{eq: based diffeo gps as colimit}
\underset{\substack{I\subset S,\,p\;\!\not\in\;\!\mathring I\,\,\,}}{\mathrm{colim^{\scriptscriptstyle H}}}\, \Diff_0(I)\,\to\, \Diff_*(S)
\quad\qquad
\underset{\substack{I\subset S,\,p\;\!\not\in\;\!\mathring I\,\,\,}}{\mathrm{colim^{\scriptscriptstyle H}}}\, \Diff_0^\R(I)\,\to\,\Diff_*^\R(S)
\end{equation}
are isomorphisms of topological groups.
\end{proposition}

\begin{proof}The proof is identical to that of Proposition~\ref{prop: colim^H for LG}. 
Let $\mathcal I$ be the poset of subintervals of $S$ whose interior does not contain $p$, and
let $N\triangleleft\, {\mathrm{colim}}_{\mathcal I}\, \Diff_0(I)$ be the normal subgroup generated by commutators of diffeomorphisms whose supports have disjoint interiors.
The proof of Theorem~\ref{thm2: Diff=colimit} applies verbatim (using a cover $\{J_j\}_{j=1...n}$ for which the $J_j\cap J_{j+1}$ have length $2d$ for some $d$, $J_0\cap J_n=\{p\}$, and all other intersections are empty)
and shows that the map $(\mathrm{colim}_{\mathcal I}\, \Diff_0(I))\big/N\,\to\,\Diff_*(S)$ is an isomorphism of topological groups.
Since $\Diff_*(S)$ is Hausdorff, the natural map
\begin{equation*}
\Big(\underset{I\subset S^1,\,p\;\!\not\in\;\!\mathring I\,\,\,}{\mathrm{colim^{\scriptscriptstyle H}}}\, \Diff_0(I)\Big)\big/N^{\scriptscriptstyle \mathrm H}\,\to\,\Diff_*(S)
\end{equation*}
is an isomorphism, where $N^{\scriptscriptstyle \mathrm H}$ denotes the image of $N$ in $\mathrm{colim}^{\scriptscriptstyle \mathrm H}_{\mathcal I}\Diff_0(I)$.
The end of the proof consists in showing that $N^{\scriptscriptstyle \mathrm H}$ is trivial.
The argument is identical to the one in Proposition~\ref{prop: colim^H for LG}.
\end{proof}

\section{Application: representations of loop group conformal nets}

Fix a compact, simple, simply connected Lie group $G$, and let $k>0$ be an integer.
Let $\g$ be the complexified Lie algebra of $G$.
There is a certain central extension $\hat\g$ of $\g[z,z^{-1}]$ called the \emph{affine Lie algebra}.
There is also a conformal net $\cA_{G,k}$ associated to $G$ and $k$, called \emph{the loop group conformal net} (we will review these notions below).

It is well believed among experts that there should be an equivalence between the category of representations of $\cA_{G,k}$ and a certain category of representations of $\hat \g$
(see Conjecture~\ref{conj: Rep(A_G,k) cong Rep^k(hat g)} for a precise statement).
In this section, leveraging Theorems~\ref{thm: LG=colimit -- BIS} and~\ref{thm: Diff=colimit -- BIS}, we prove one half of this conjecture.
Namely, given a representation of the loop group conformal net, we construct a representation of the corresponding affine Lie algebra:
\[
\Rep(\cA_{G,k}) \,\,\,\rightsquigarrow\,\,\, \Rep(\hat \g).
\]

\subsection{Loop group conformal nets and affine Lie algebras}
\subsubsection{Loop group conformal nets}
Let $S$ be a circle.
Let $\widetilde{LG}$ be the basic central extension of $LG=\mathrm{Map}(S,G)$, and let $\widetilde{LG}^k$ be its quotient by the central subgroup $\mu_k\subset U(1)$ of $k$-th roots of unity.
Upon identifying $U(1)/\mu_k$ with $U(1)$, we get a central extension
\[
\tikzmath[xscale=1.5]{
\node (0) at (0,0) {$0$};
\node (1) at (1,0) {$U(1)$};
\node (2) at (2.07,.1) {$\widetilde{LG}^k$};
\node (3) at (3,0) {$LG$};
\node (4) at (4,0) {$0$};
\draw [->] (0) -- (1);
\draw [->] (1) -- ++(.75,0);
\draw [<-] (3) -- ++(-.7,0);
\draw [->] (3) -- (4);
}
\]
called the level $k$ central extension of $LG$.
Given an interval $I\subset S$, we write $\widetilde{L_IG}^k$ for the restriction of that central extension to the subgroup $L_IG$:
\begin{equation}\label{eq: central ext of L_IG}
\tikzmath[xscale=1.5]{
\node (0) at (0,0) {$0$};
\node (1) at (1,0) {$U(1)$};
\node (2) at (2.15,.1) {$\widetilde{L_IG}^k$};
\node (3) at (3.15,0) {$L_IG$};
\node (4) at (4.15,0) {$0$.};
\draw [->] (0) -- (1);
\draw [->] (1) -- ++(.75,0);
\draw [<-] (3) -- ++(-.7,0);
\draw [->] (3) -- (4);
}
\end{equation}
The latter only depends on $I$ and not on the choice of circle $S$ in which the interval is embedded.

Recall that the group algebra $\C[\mathcal G]$ of a group $\mathcal G$ is the set of finite linear combinations of elements of $\mathcal G$.
We write $[g]\in\C[\mathcal G]$ for the image in the group algebra of an element $g\in \mathcal G$.
Given a central extension $0\to U(1)\to \widetilde{\mathcal G}\to \mathcal G\to 0$, the \emph{twisted group algebra} of $\mathcal G$ is the quotient of $\C[{\;\!\widetilde{\mathcal G}\;\!}]$ by the relation
$[\lambda g]\sim \lambda [g]$, for $\lambda\in U(1)$ and $g\in \widetilde{\mathcal G}$.

Following \cite[\S1.A]{BDH-cn1}, a conformal net is a functor from the category of intervals and embeddings to the category of von Neumann algebras and $*$-algebra homomorphisms (satisfying various axioms).
Associated to each compact, simple, simply connected Lie group $G$ and integer $k\ge 1$,
there is a conformal net $\cA_{G,k}$ called the \emph{loop group conformal net}.
The loop group conformal net sends an interval $I$ to a certain von Neumann algebra $\cA_{G,k}(I)$.
The latter is a completion of the twisted group algebra of $L_IG$ associated to the central extension \eqref{eq: central ext of L_IG}.
In particular, there is a homomorphism
\begin{equation}\label{eq: LIG --> U(A)}
\,\widetilde{L_IG}^k\,\to\,\, U\big(\cA_{G,k}(I)\big)
\end{equation}
from $\widetilde{L_IG}^k$ to the group of unitaries of $\cA_{G,k}(I)$.\vspace{-1mm}
The homomorphism \eqref{eq: LIG --> U(A)} is continuous for the topology on $\widetilde{L_IG}^k$ induced from the $C^\infty$ topology on $L_IG$, and the strong operator topology on $U(\cA_{G,k}(I))$.
We refer the reader to \cite[\S4.C]{BDH-cn1}\cite{MR1231644}\cite[\S8]{WZW-classification}\cite{Toledano(PhD-thesis)}\cite{MR1645078} for background on loop group conformal nets.

We write $B(H)$ for the algebra of bounded operators on a Hilbert space $H$, and $U(H)$ for the group of unitary operators.

\begin{definition}
A representation of a conformal net $\cA$ on a Hilbert space $H$ is a collection of actions (i.e. normal, unital, $*$-homomorphisms) $\rho_I:\cA(I)\to B(H)$, indexed by the subintervals $I\subset S^1$, which are compatible in the sense that $\rho_I|_{\cA(J)}=\rho_J$ for every $J\subset I\subset S^1$.
\end{definition}

We write $\Rep(\cA)$ for the category of representation of $\cA$, and 
$\Rep_{\mathrm{f}}(\cA)$ for the subcategory whose objects are finite direct sums of irreducible representations.

\subsubsection{Affine Lie algebras}

Let $\g_\R$ be the Lie algebra of $G$, and let $\g:=\g_\R\otimes_\R\C$ be its complexification.
The \emph{affine Lie algebra} $\hat\g$ is the central extension of $\g[z,z^{-1}]$ by the $2$-cocycle $(f,g)\mapsto \mathrm{Res}_{z=0}\langle f,dg\rangle$,
where $\langle\,\,,\,\rangle$ denotes the basic inner product on $\g$ (c.f. Section~\ref{sec:The free loop group}).\footnote{
This normalization ensures that the set of possible levels of integrable positive energy representations of $\hat\g$ is exactly the set $\N$ of non-negative integers.
}
The \emph{Kac-Moody algebra} is the semi-direct product $\C\,\ltimes\, \hat\g$ associated to the derivation
$(f,a)\mapsto (-z\frac{\partial f}{\partial z},0)$ of $\hat\g$.
We write $\ell_0$ for the generator of $\C$ in the semi-direct product.

\begin{definition}[{\cite[Chapt.\,3 and 10]{MR1104219}}]\label{def: Rep Lg}
Let $k\in\N$.
A representation $\rho$ of $\hat\g$ on a vector space $V$ is called a level~$k$ integrable positive energy representation if:
\begin{enumerate}
\item It is the restriction of a representation of $\C\ltimes \hat\g$ for which the generator $L_0:=\rho(\ell_0)$ of $\C$ is diagonalizable, with positive spectrum.
\item For every nilpotent element $X\in\g$ and every $n\ge 0$, the operator $Xz^{-n}$ acts locally nilpotently on $V$
(the operators $Xz^n$ are automatically locally nilpotent for $n>0$).
\item The central element $1\in\C\subset \hat \g$ acts by the scalar $k$.
\end{enumerate}
We note that the choice of operator $L_0$ is not part of the data of an integrable positive energy representation.
\end{definition}

We write $\Rep^k(\hat \g)$ for the category of level~$k$ integrable positive energy representations of $\hat\g$,
and write $\Rep^k_{\mathrm{f}}(\hat \g)$ for the subcategory whose objects are finite sums of irreducible representations.

It is well known that every object of $\Rep^k(\hat \g)$ can be equipped with a positive definite $\hat\g$-invariant inner product \cite[Chapt.\,11]{MR1104219}, and thus completed to a Hilbert space.
The action of $\hat\g$ extends to an action on the Hilbert space by unbounded operators, and
the real form $\big(\g[z,z^{-1}]\cap C^\infty(S^1,\g_\R)\big)\oplus i\R\subset \hat \g$ acts by skew-adjoint operators.
The latter can then be integrated to a strongly continuous positive energy unitary representation of $\widetilde{LG}^k$ \cite{MR733047, MR1674631}.
(A unitary representation $\mathcal G\to U(H)$ is called strongly continuous if it is continuous with respect to the strong operator topology on~$U(H)$.)\vspace{-.7mm}
In order to have a better parallel with Definition~\ref{def: Rep Lg}, we\vspace{-.7mm} prefer to view representations of $\widetilde{LG}^k$ as representations of $\widetilde{LG}$, via the quotient map $\widetilde{LG}\to \widetilde{LG}^k$:

\begin{definition}[{\cite[\S 9.2]{MR900587}}]\label{def: pos en rep of LG}
A strongly continuous unitary representation $\rho:\widetilde{LG}\to U(H)$ is called a level~$k$ positive energy representation if:
\begin{enumerate}
\item It is the restriction of a representation $S^1\ltimes \widetilde{LG}\to U(H)$, for which the generator $L_0:=-i\tfrac{d}{dt}\big|_{t=0}(\rho(e^{it},1))$ of $S^1$ has positive spectrum.
\item $\lambda\in U(1)\subset \widetilde{LG}$ acts by scalar multiplication by $\lambda^k$.
\end{enumerate}
We note that the action of $S^1\subset S^1\ltimes \widetilde{LG}$ is not part of the data of a level $k$ positive energy representation.
\end{definition}

We write $\Rep^k(LG)$ for the category of level~$k$ positive energy representations of $\widetilde{LG}$,
and $\Rep^k_{\mathrm{f}}(LG)$ for the subcategory whose objects are finite direct sums of irreducible representations.

\subsection{The comparison functors\\ $\Rep(\cA_{G,k}) \to \Rep^k(LG)$ and $\Rep^k(LG) \to \Rep^k(\hat{\g})$}

It has long been expected that there should be a one-to-one correspondence between representations of $\cA_{G,k}$ and level $k$ integrable positive energy representations of $\hat\g$.
One way to state this is as an equivalence of categories:
\[
\Rep_{\mathrm{f}}(\cA_{G,k})\,\cong\, \Rep^k_{\mathrm{f}}(\hat \g).
\]
We prefer the following statement, as it excludes the possibility of $\cA_{G,k}$ having representations which are not direct sums of irreducible ones:

\begin{conjecture}\label{conj: Rep(A_G,k) cong Rep^k(hat g)}
Let $\mathsf{Vec}_{\mathrm{f}}$ be the category of finite dimensional vector spaces, and let $\mathsf {Hilb}$ be the category of Hilbert spaces and bounded linear maps.
Then there is a natural equivalence of categories:
\[
\Rep(\cA_{G,k}) \,\,\cong\,\, \Rep^k_{\mathrm{f}}(\hat \g)\otimes_{\mathsf{Vec}_{\mathrm{f}}} \mathsf {Hilb}.
\]
\end{conjecture}

\noindent
Here, the objects of $\Rep^k_{\mathrm{f}}(\hat \g)\otimes_{\mathsf{Vec}_{\mathrm{f}}} \mathsf {Hilb}$ are formal expressions of the form $\bigoplus_{i=1}^nV_i\otimes H_i$ with $V_i\in
\Rep^k_{\mathrm{f}}(\hat \g)$ and $H_i\in\mathsf {Hilb}$, and
$\mathrm{Hom}(V_i\otimes H_i,V'_j\otimes H'_j)=\mathrm{Hom}_{\Rep^k_{\mathrm{f}}(\hat \g)}(V_i,V'_j) \otimes_\C \mathrm{Hom}_{\mathsf {Hilb}}(H_i,H'_j)$.

\begin{remark}\label{rem: ok for SU(n)}
For $G = SU(n)$, Conjecture~\ref{conj: Rep(A_G,k) cong Rep^k(hat g)} is a consequence of \cite[Thm. 2.2]{MR1776984}, \cite[Thm. 33]{MR1838752}, and \cite[42, Thm. 3.5]{MR1776984}.
\end{remark}

Theorems \ref{thm: Rep(A) --> Rep(LG)} and \ref{thm: Rep(LG) --> Rep(hat g)} below, together with Remark~\ref{rem Zellner did it} (or Remark~\ref{rem: ok for SU(n)}), prove one half of Conjecture~\ref{conj: Rep(A_G,k) cong Rep^k(hat g)}.
Namely, they combine to a fully faithful functor
\begin{equation}\label{eq: Rep(A)-->Rep(g hat)}
\Rep(\cA_{G,k}) \,\,\to\,\, \Rep^k_{\mathrm{f}}(\hat \g)\otimes_{\mathsf{Vec}_{\mathrm{f}}} \mathsf {Hilb}.
\end{equation}
This proves, among other things, that every representation of $\cA_{G,k}$ is a direct sum of irreducible ones,
and that there is an injective map (conjecturally a bijection) from the set of isomorphism classes of 
irreducible objects of $\Rep(\cA_{G,k})$ to the set of isomorphism classes of irreducible objects of $\Rep^k(\hat \g)$.

\begin{remark}An alternative proof of the above result can be found in the unpublished manuscript \cite{Local_energy_bounds}.
\end{remark}

\begin{remark}
Constructing the inverse functor of \eqref{eq: Rep(A)-->Rep(g hat)}
requires something called ``local equivalence''.
Results about local equivalence can be found in \cite[Thm.\,B in \S17]{MR1645078} and \cite[Prop.\,2.4.1 in Chapt.\,IV]{Toledano(PhD-thesis)}.
The unpublished preprint \cite[\S15]{Wass_unpub_1990} seems to contain most of the ingredients of a proof, but falls short of being a complete argument.
\end{remark}

\begin{remark}
By the results in \cite[App.\,D]{MR1838752}, the existence of an injective map from the set of isomorphism classes of irreducible $\cA_{G,k}$-reps to the set of
isomorphism classes of irreducible positive energy $\hat\g$-reps implies that every $\cA_{G,k}$-rep is a direct sum of irreducible ones.
This fact could have been used to simplify the proof of Theorem~\ref{thm: Rep(A) --> Rep(LG)} (specifically the direct integral argument) by assuming from the beginning that $L_0$ is diagonalizable. We have opted instead for a more self-contained exposition.
\end{remark}

\begin{remark}
The corresponding questions for the Virasoro conformal net have been studied by a number of people.
Carpi \cite{MR2031030}, based on results by Loke \cite{Loke(PhD-thesis)}
and D'Antoni and K\"oster \cite{MR2078164},
proved that every irreducible positive energy representation of a Virasoro conformal net
comes from a positive energy representation of the Virasoro algebra with same central charge.
The converse (local normality) was shown to hold by Weiner \cite{MR3627412} (with partial results by Buchholz and Schulz-Mirbach \cite{MR1079298}), and the positive energy condition was removed in \cite{MR2231680}.
\end{remark}

\begin{theorem}\label{thm: Rep(A) --> Rep(LG)}
For\vspace{-1mm} every representation of the loop group conformal net $\cA_{G,k}$,
the actions of the subgroups $\widetilde{L_IG}^k\subset U(\cA_{G,k}(I))$
assemble to a level $k$ positive energy representation of the loop group.
That construction yields a fully faithful functor
\begin{equation}\label{eq: Rep(A_G,k) --> Rep^k(LG)}
\Rep(\cA_{G,k}) \,\,\to\,\, \Rep^k(LG).
\end{equation}
\end{theorem}

\begin{proof}
Let $H$ be a representation of $\cA_{G,k}$.
By definition, $H$ is equipped with compatible actions $\cA_{G,k}(I)\to B(H)$ for all $I\subset S^1$.\vspace{-1mm}
Precomposing by the maps \eqref{eq: LIG --> U(A)}, we get a compatible system of homomorphisms $\widetilde{L_IG}^k\to U(H)$.
By Theorem~\ref{thm: LG=colimit -- BIS}, these assemble to a strongly continuous action
\begin{equation}\label{eq: assembled action}
\widetilde{LG}^k=\mathrm{colim}_{I\subset S^1}\widetilde{L_IG}^k\to\, U(H).
\end{equation}
By construction, any $\lambda\in U(1)\subset \widetilde{LG}^k$ acts by scalar multiplication by $\lambda$.

By the diffeomorphism covariance of the loop group conformal nets (\cite[Prop.\,4.3]{BDH-cn1} with \cite[Thm. 6.7]{MR733047} or \cite[Thm. 6.1.2]{MR1674631}),\vspace{-1mm}
there exist canonical maps $\Diff_0^\R(I)\to \cA_{G,k}(I)$ which assemble to homomorphisms $\Diff_0^\R(I)\ltimes \widetilde{L_IG}^k\! \to U(\cA_{G,k}(I))$.\vspace{-1mm}
Composing with the map to $U(H)$ and using \eqref{eq: assembled action}, we get a compatible system of maps
$\Diff_0^\R(I)\ltimes \widetilde{LG}^k\! \to U(H)$.
By Theorem~\ref{thm: Diff=colimit -- BIS}, these then assemble to a strongly continuous action
\begin{equation*}
\Diff_+^{\R\times \Z}(S^1)\ltimes \widetilde{LG}^k\to\, U(H).
\end{equation*}
Precomposing by the quotient map $\widetilde{LG}\to \widetilde{LG}^k$, we get an action of $\Diff_+^{\R\times \Z}(S^1)\ltimes \widetilde{LG}$ such that
every $\lambda\in U(1)\subset \widetilde{LG}$ acts by scalar multiplication by $\lambda^k$.
In particular, we get an action of $(S^1)^\Z\ltimes \widetilde{LG}$ on $H$, where $(S^1)^\Z\subset \Diff_+^{\R\times \Z}(S^1)$ denotes the universal cover of $S^1\subset \Diff_+^{\R}(S^1)$.
The main result of \cite{MR2231680} shows that the generator $L_0$ of $(S^1)^\Z$ has positive spectrum.

So far, we have constructed a representation of $\widetilde{LG}$ on $H$ that satisfies all the conditions of a level $k$ positive energy representation,
except that the $S^1$ is replaced by its universal cover $(S^1)^\Z$.
In order to show that $H$ is a positive energy representation (Definition~\ref{def: pos en rep of LG}), we need to modify the action of $(S^1)^\Z$ so that it descends to an action of $S^1$.

Decompose $H$ as a direct integral according to the characters of the central $\Z\subset (S^1)^\Z$:
\[
H=\int^\oplus_{\theta\in U(1)} H_\theta.
\]
(This direct integral will turn out to be a mere direct sum, but don't know this at the moment.)
Direct integrals for loop group representations are tricky, because disintegration theory only applies to separable locally compact groups, and $\widetilde{LG}$ is not locally compact.
So we proceed with care.
In particular, we never make the claim that the Hilbert spaces $H_\theta$ carry actions of $\widetilde{LG}$.

For each $\theta\in U(1)$, extend the character $n\mapsto \theta^n$ of $\Z$ to a character $z\mapsto z^{\log(\theta)/2\pi i}$ of $(S^1)^\Z$ (principal branch of the logarithm).
Let $\C_\theta$ denote the vector space $\C$, equipped with the action of~$(S^1)^\Z$ given by the above character.
Then the representation
\[
H':=\int^\oplus_{\theta\in U(1)} H_\theta\otimes \overline{\C_{\theta}}
\]
of $(S^1)^\Z$ descends to a representation of $S^1$ whose generator has positive spectrum (the spectrum of $L_0$ has been modified by a bounded amount).
As mere vector spaces, we have $\overline{\C_{\theta}}\cong \C$, and therefore $H'\cong H$.
Use this isomorphism to equip $H'$ with an action of $\widetilde{LG}$.
We wish to show that the actions of $S^1$ and of $\widetilde{LG}$ on $H'$ assemble to an action of $S^1\ltimes\widetilde{LG}$.

Pick a countable dense subgroup $(S^1)^\Z_\star\subset (S^1)^\Z$ that contains the central $\Z$,
and let $S^1_\star:=(S^1)^\Z_\star/\Z\subset S^1$.
Pick an $(S^1)^\Z_\star$-invariant countable dense subgroup $\widetilde{LG}_\star\subset\widetilde{LG}$.
Since $\Z$ is central in $(S^1)^\Z_\star\ltimes\widetilde{LG}_\star$, the Hilbert spaces $H_\theta$
carry representations of $(S^1)^\Z_\star\ltimes\widetilde{LG}_\star$ for almost all $\theta$.
By construction, on almost each $H_\theta\otimes \overline{\C_{\theta}}$, the action of $(S^1)^\Z_\star\ltimes\widetilde{LG}_\star$ descends to an action of $S^1_\star\ltimes\widetilde{LG}_\star$.
The actions of $S^1_\star$ and $\widetilde{LG}_\star$ on $H'$ therefore assemble to an action of $S^1_\star\ltimes\widetilde{LG}_\star$.
At last, since $S^1_\star\ltimes\widetilde{LG}_\star$ is dense in $S^1\ltimes\widetilde{LG}$ and since the actions of $S^1$ and $\widetilde{LG}$ on $H'$ are strongly continuous, 
these two actions assemble to an action of $S^1\ltimes\widetilde{LG}$.
This finishes the proof that $H'$, and hence $H$, is a positive energy representation of $\widetilde{LG}$.
\end{proof}

Positive energy representations of $S^1\ltimes \widetilde{LG}$ come with no smoothness assumptions.
It is therefore not clear, a priori, that it should be possible to differentiate them.
Zellner showed that, in such a representation, the set of smooth vectors is always dense \cite[Thm\,2.16]{Zellner}.
One can therefore differentiate it to a representation of the corresponding Kac-Moody Lie algebra.
We present an alternative proof of that same result.
(Our proof does not cover the case $G=SU(2)$:
it relies on the fact that every rank $2$ sub diagram of the affine Dynkin diagram of $G$ is of finite type, something which holds for all groups except for $SU(2)$.)

\begin{theorem}\label{thm: Rep(LG) --> Rep(hat g)}
Let $G\not = SU(2)$.
A level $k$ positive energy representation $\widetilde{LG}\to U(H)$ can be differentiated to a level $k$ integrable positive energy representation of $\hat \g$ on a dense subset of~$H$.
This construction yields an equivalence of categories
\[
\Rep^k(LG) \,\,\to\,\, \Rep^k_{\mathrm{f}}(\hat \g)\otimes_{\mathsf{Vec}_{\mathrm{f}}} \mathsf {Hilb}.
\]
\end{theorem}

\begin{remark}\label{rem Zellner did it}
For $G=SU(2)$ (and indeed for any compact simple group Lie group $G$), Theorem~\ref{thm: Rep(LG) --> Rep(hat g)} follows from \cite[Thm\,2.16]{Zellner}.
\end{remark}

\begin{proof}
An integration functor $\Rep^k_{\mathrm{f}}(\hat \g) \to \Rep^k(LG)$ was constructed in \cite{MR733047} and \cite{MR1674631}.
The functor sends irreducible representations of $\hat \g$ to irreducible representations of $\widetilde{LG}$. It is therefore visibly fully faithful.

The category $\Rep^k(LG)$ is tensored over the category of Hilbert spaces (i.e., the tensor product of a positive energy $\widetilde{LG}$ representation with a Hilbert space is again a positive energy $\widetilde{LG}$ representation). So the above functor extends to a functor
\begin{equation}\label{eq: differentiation functor}
\Rep^k_{\mathrm{f}}(\hat \g)\otimes_{\mathsf{Vec}_{\mathrm{f}}} \mathsf {Hilb} \to \Rep^k(LG),
\end{equation}
which is again visibly fully faithful.
In order to show that the functor \eqref{eq: differentiation functor} is an equivalence of categories, we need to show that it is essentially surjective.

Let $T_G$, $T_{\widetilde{LG}}$, $T_{S^1\ltimes\widetilde{LG}}$ be the maximal tori of $G$, $\widetilde{LG}$, and $S^1\ltimes\widetilde{LG}$, so that
\[
T_{\widetilde{LG}}=U(1)\times T_G
\qquad\text{and}\qquad
T_{S^1\ltimes\widetilde{LG}}=S^1\times U(1)\times T_G.
\]
Let $\Lambda_G$, $\Lambda_{\widetilde{LG}}$, $\Lambda_{S^1\ltimes\widetilde{LG}}$ be the character lattices of $T_G$, $T_{\widetilde{LG}}$, and $T_{S^1\ltimes\widetilde{LG}}$,
and let
\[
\Lambda_k:= \big\{\chi\in \Lambda_{S^1\ltimes\widetilde{LG}}\,\big|\,\chi(\lambda)=\lambda^k\,\,\, \text{for every } \lambda\in U(1)\big\}
\]
(an affine sublattice canonically isomorphic to $\Z\times \Lambda_G$).
We write $\Lambda_k^+\subset \Lambda_k$ for the set of possible highest weights of irreducible highest weight level $k$ integrable representations of $\C\ltimes \hat\g$ \cite[Chapt.\,10]{MR1104219}.
Let $\pi:\Lambda_k\to \Lambda_G$ be the projection map, and let $\mathrm{A}_k:=\pi(\Lambda_k^+)\subset \Lambda_G$.
The finite set $\mathrm{A}_k$ parametrizes the isomorphism classes of irreducible objects of $\Rep^k(\hat \g)$.

Let $H$ be a level $k$ positive energy representation of $\widetilde{LG}$.
By definition, the action of $\widetilde{LG}$ extends (in a non-unique way) to an action of $S^1\ltimes\widetilde{LG}$
such that the generator of $S^1$ has positive spectrum.
Pick such as extension of the action.
Then $H$ decomposes as the Hilbert space direct sum of its weight spaces:
\[
H={\bigoplus_{\lambda\in\Lambda_k}}^{\tikz{
\useasboundingbox (0,-.1) rectangle (.15,.1);
\node[scale=.9] at (0,0){$\scriptstyle\ell_2$};}} H_\lambda.
\]
Let $P:=\{\lambda\in \Lambda_k\,|\,H_\lambda\not =0\}$.
The affine Weyl group $\widehat W= NT_{S^1\ltimes\widetilde{LG}}/T_{S^1\ltimes\widetilde{LG}}$ acts on $\Lambda_k$, and preserves $P$.
By the positive energy condition, $P$ is contained in the ``half-space'' $\Z_{\ge 0}\times \Lambda_G\subset \Z\times \Lambda_G = \Lambda_k$.
Combining this with its $\widehat W$-invariance, we learn that $P$ is contained in a paraboloid.

Let $D$ be the affine Dynkin diagram associated to $G$.
Every node $a\in D$ corresponds to an embedding $SU(2)\hookrightarrow \widetilde{LG}$.
We write $SU(2)_a$ for the subgroup of $\widetilde{LG}$ which is the image of that embedding.
Let $T_a\subset T_{S^1\ltimes\widetilde{LG}}$ be the subgroup of the torus which centralizes  $SU(2)_a$,
and let $\Lambda_a$ be the corresponding quotient of $\Lambda_{S^1\ltimes\widetilde{LG}}$, with projection map $p_a:\Lambda_{S^1\ltimes\widetilde{LG}}\to \Lambda_a$
($\Lambda_a$ is the character lattice of $T_a$). The kernel of $p_a$ has rank one.
Let us also define $\Lambda_{k,a}:=p_a(\Lambda_k)$.

Since $P$ is contained inside a paraboloid,
for each $\sigma\in \Lambda_{k,a}$, the set $\Lambda(\sigma):=\{\lambda\in P\,|\,p_a(\lambda)=\sigma\}$ is finite.
It follows that, for every $\sigma$, the representation of $SU(2)_a$ on $\bigoplus_{\lambda\in\Lambda(\sigma)} H_\lambda$ contains only finitely many isomorphism classes of irreducible $SU(2)_a$ representations. In particular, the action of $\mathfrak{su}(2)_a$ on $\bigoplus_{\lambda\in\Lambda(\sigma)} H_\lambda$ is by bounded operators (which are in particular everywhere defined).
In this way, we obtain actions of the Lie algebras $\mathfrak{su}(2)_a$ on the algebraic direct sum
\[
\check H\,:=\,\bigoplus_{\lambda\in\Lambda_k} H_\lambda\,\subset\, H.
\]
Those Lie algebras contain all the generators $\{E_a, F_a, H_a\}$ of the Serre presentation of~$\hat\g$.

To check that the above generators satisfy the Serre relations, we consider rank two subgroups of $\widetilde{LG}$.
For every pairs of vertices $a,b\in D$, the subgroup $G_{ab}\subset\widetilde{LG}$ generated by $SU(2)_a$ and $SU(2)_b$ is compact
as it correspnds to the sub-Dynkin diagram of $D$ on the two vertices, and the latter is either $A_1\sqcup A_1$, $A_2$, $B_2$, or $G_2$ 
--- this is where we use that $G$ is not $SU(2)$.
Applying the same arguments as above, we see that there are actions of the corresponding Lie algebras $\mathfrak g_{ab}$ on $\check H$.
Every Serre relation is detected in one of the Lie algebras $\mathfrak g_{ab}$.
So the generators $\{E_a, F_a, H_a\}$ satisfy all the relations and we get an action of $\hat \g$ on $\check H$.
Finally, we can use the action of $T_{S^1\ltimes\widetilde{LG}}$ on $H$ (and thus on $\check H$) to extend the action of $\hat\g$ on $\check H$ to a action of $\C\ltimes \hat\g$.

Let $L_\mu$ be the irreducible highest weight representation of $\C\ltimes \hat\g$ with highest weight $\mu\in \Lambda_k^+$.
We write $H(\mu):=\mathrm{Hom}_{\C\ltimes \hat\g}(L_\mu,\check H)$ for the multiplicity space of $L_\mu$ inside $\check H$, so that
\begin{equation}\label{eq: bigoplus for H_lambda}
\check H\,=\,\bigoplus_{\mu\in\Lambda_k^+}\, H(\mu)\otimes L_\mu.
\end{equation}
($\check H$ is a representation of $\C\ltimes \hat\g$ satisfying the three conditions listed in Definition~\ref{def: Rep Lg}, and the category of such representations is semi-simple in the sense that every object is a direct sum of irreducible ones \cite[Chapt.\,9,\,10]{MR1104219}.)
The multiplicity space $H(\mu)$ can also be described as the joint kernel of the lowering operators $F_a$ acting $H_\mu$.
By this second description, we see that $H(\mu)$ is a closed subspace of $H_\mu$, and thus a Hilbert space in its own right.
Letting $\bar L_\mu$ be the Hilbert space completion of $L_\mu$,
we can then upgrade the isomorphism \eqref{eq: bigoplus for H_lambda} to an isomorphism of Hilbert spaces:
\begin{equation} \label{eq: bigoplus for H_lambda -- BIS}
H\,=\,{\bigoplus_{\mu\in\Lambda_k^+}}^{\tikz{
\useasboundingbox (0,-.1) rectangle (.15,.1);
\node[scale=.9] at (0,0){$\scriptstyle\ell_2$};}}\, H(\mu)
\otimes \bar L_\mu
\end{equation}
(where $\otimes$ now denotes the Hilbert space tensor product).

Recall the projection $\pi:\Lambda_k^+\to \mathrm{A}_k$.
Two representations $L_\mu$ and $L_{\mu'}$ of $\C\ltimes \hat \g$ are isomorphic as representations of $\hat \g$ if and only if $\pi(\mu)=\pi(\mu')$.
For $\lambda\in\mathrm{A}_k$, let $H[\lambda]:=\bigoplus_{\pi(\mu)=\lambda}^{\tikz{
\useasboundingbox (-.13,-.1) rectangle (.15,.1);
\node[scale=.9] at (0,0){$\scriptstyle\ell_2$};}} H(\mu)$.
The decomposition \eqref{eq: bigoplus for H_lambda -- BIS} then induces a direct sum decomposition
\[
H\;=\;{\bigoplus_{\lambda\in\mathrm{A}_k}}\,\, H[\lambda]
\otimes \bar L_\lambda.
\]
This finishes the proof that $H$ is in the essential image of the functor \eqref{eq: differentiation functor}.
\end{proof}

\subsection{The based loop group and its representations}

Let $\Rep^k_{\scriptscriptstyle \mathrm{l.n.\!\!}}(LG)$ be the essential image of the functor \eqref{eq: Rep(A_G,k) --> Rep^k(LG)}.
We call it the category of \emph{locally normal representations of $\widetilde{LG}$ at level $k$}.
By Theorem~\ref{thm: Rep(LG) --> Rep(hat g)} (and Remark~\ref{rem Zellner did it}),
Conjecture~\ref{conj: Rep(A_G,k) cong Rep^k(hat g)} is equivalent to the statement that $\Rep^k_{\scriptscriptstyle \mathrm{l.n.\!\!}}(LG)=\Rep^k(LG)$ (the latter was defined in Definition~\ref{def: pos en rep of LG}).

In \cite{CS(pt)}, we introduced the category of \emph{locally normal representations}\footnote{In the first versions of \cite{CS(pt)}, we called these representations ``positive energy representations''.
We call them here `locally normal representations' and
reserve the term `positive energy representations' for another type of representations
(we conjecture that the two conditions are equivalent --- see Conjecture~\ref{conj: stronger conjecture}).} of $\widetilde{\Omega G}$ at level $k$ (Definition~\ref{def: locally normal rep of Omega G}).
We denote it here by $\Rep^k_{\scriptscriptstyle \mathrm{l.n.\!\!}}(\Omega G)$.
In that same preprint, we announced that the Drinfel'd center of 
the category of locally normal representations of the based loop group at level $k$ is equivalent to 
the category of locally normal representations of the free loop group at level $k$,
where the latter is equipped with the fusion and braiding inherited from $\Rep(\cA_{G,k})$:
\begin{equation}\label{eq: Z(Rep^k(Omega G))Rep^k(LG)}
Z\big(\Rep^k_{\scriptscriptstyle \mathrm{l.n.\!\!}}(\Omega G)\big)\,=\,\Rep^k_{\scriptscriptstyle \mathrm{l.n.\!\!}}(LG).
\end{equation}
In the more recent preprint \cite{Bicommutant-categories-from-conformal-nets}, we considered the category $T_{\cA_{G,k}}$ of solitons of the conformal net $\cA_{G,k}$, and proved that
\[
Z(T_{\cA_{G,k}})=\Rep^k_{\scriptscriptstyle \mathrm{l.n.\!\!}}(LG).
\]
In order to complete our proof of \eqref{eq: Z(Rep^k(Omega G))Rep^k(LG)}, we need to identify $\Rep^k_{\scriptscriptstyle \mathrm{l.n.\!\!}}(\Omega G)$ with $T_{\cA_{G,k}}$.
This is the main result of the present section.

\subsubsection{Solitons and representations of $\Omega G$}\label{sec: Solitons and representations of OmG}

We work with the standard circle $S^1$, and the base point $p:=1\in S^1$.

\begin{definition}\label{def: locally normal rep of Omega G}
A strongly continuous representation $\rho:\widetilde{\Omega G}\to U(H)$ is a locally normal level~$k$ representation if
\begin{enumerate}
\item
$\lambda\in U(1)\subset \widetilde{\Omega G}$ acts by scalar multiplication by $\lambda^k$, and
\item
for every interval $I \subset S^1$, $p \not \in \mathring I$, the action of $\widetilde{L_IG}$ extends to an action of the von Neumann algebra $\cA_{G,k}(I)$.
\end{enumerate}
We write $\Rep^k_{\scriptscriptstyle \mathrm{l.n.\!\!}}(\Omega G)$ for the category of locally normal representations of $\widetilde{\Omega G}$ at level~$k$.
\end{definition}

For every conformal net $\cA$, we also have its category of solitons \cite{MR1652746, MR1892455, MR1332979, MR2100058}:

\begin{definition}
A soliton (or solitonic representation) of a conformal net $\cA$ on a Hilbert space $H$ is a collection of actions $\rho_I:\cA(I)\to B(H)$
for every interval $I\subset S^1$ with $p \not \in \mathring I$ which satisfy $\rho_I|_{\cA(J)}=\rho_J$ for every $J\subset I\subset S^1$.
\end{definition}

There is an obvious fully faithful functor
\begin{equation}\label{eq: Rep^k(Omega G) --> eq: Sol(A_G,k)}
\Rep^k_{\scriptscriptstyle \mathrm{l.n.\!\!}}(\Omega G)\to T_{\cA_{G,k}}
\end{equation}
which takes a locally normal representation of $\widetilde{\Omega G}$
and only remembers the actions of the von Neumann algebras $\cA_{G,k}(I)$, $p \not \in \mathring I$.

\begin{theorem}
The functor \eqref{eq: Rep^k(Omega G) --> eq: Sol(A_G,k)} is an equivalence of categories.
\end{theorem}

\begin{proof}
We need to show that the functor is essentially surjective.
Let $H$ be a soliton of $\cA_{G,k}$.
By definition, $H$ is equipped with compatible actions $\cA_{G,k}(I)\to B(H)$ for all intervals $I\subset S^1$, $p \not \in \mathring I$.\vspace{-1mm}
Precomposing by the maps \eqref{eq: LIG --> U(A)}, we get homomorphisms $\widetilde{L_IG}^k\to U(H)$.
By \eqref{eq: based loop gps as colimit -- BIS}, and using that $U(H)$ is Hausdorff,
these assemble to a strongly continuous action
\begin{equation*}
\widetilde{\Omega G}^k=\underset{I\subset S^1,\,p\;\!\not\in\;\!\mathring I\,\,\,}{\mathrm{colim^{\scriptscriptstyle H}}} \widetilde{L_IG}^k
\,\to\,\, U(H).
\end{equation*}
Clearly, any $\lambda\in U(1)\subset \widetilde{\Omega G}^k$ acts by scalar multiplication by $\lambda$.\vspace{-.7mm}
Precomposing by the quotient map $\widetilde{\Omega G}\to \widetilde{\Omega G}^k$, we get an action of $\widetilde{\Omega G}$ such that each
$\lambda\in U(1)\subset \widetilde{\Omega G}$ acts by $\lambda^k$.
By construction, this is a locally normal representation.
\end{proof}

Solitonic representations are also equipped with a natural action of the based diffeomorphism group.
Let $H$ be a soliton of $\cA_{G,k}$.\vspace{-1mm}
Recall that, by diffeomorphism covariance,\vspace{-1mm}
there exist homomorphisms $\Diff_0^\R(I)\ltimes \widetilde{L_IG}^k\! \to U(\cA_{G,k}(I))$ for every interval $I\subset S^1$ with $p \not \in \mathring I$.
Composing with the projection $\widetilde{L_IG}\to \widetilde{L_IG}^k$ and with the action $U(\cA_{G,k}(I))\to U(H)$, we get a compatible family of homomorphisms
\[
\Diff_0^\R(I)\ltimes \widetilde{L_IG} \to U(H).
\]
By Propositions~\ref{prop: colim^H for LG} and~\ref{prop: based Diff colimit}, these assemble to a strongly continuous action
\begin{equation} \label{eq: assembled action + Diff}
\Diff_*^\R(S^1)\ltimes \widetilde{\Omega G}\,=\!
\underset{I\subset S^1,\,p\;\!\not\in\;\!\mathring I\,\,\,}{\mathrm{colim^{\scriptscriptstyle H}}}
\Diff_0^\R(I)\ltimes \widetilde{L_IG}
\,\to\, U(H).
\end{equation}

Based on results of Carpi and Weiner \cite{MR2166846, MR2231680},
it was proved in \cite[App.\,A]{arXiv:1811.04501} that the maps $\Diff_0^\R(I) \to \cA_{G,k}(I)$ extend to a certain larger group involving non-smooth diffeomorphisms.
Given an interval $I$, let $\Diff_{1,\mathrm{ps}}(I)$ be the group of orientation preserving piecewise smooth $C^1$ diffeomorphisms of $I$ whose derivative is $1$ at the boundary points.
And let $\Diff_{1,\mathrm{ps}}(S^1)$ be the group of orientation preserving piecewise smooth $C^1$ diffeomorphisms of $S^1$ that fix the base point $p$, and whose derivative is $1$ at that point.
Let $\Diff_{1,\mathrm{ps}}^\R(I)$ and $\Diff_{1,\mathrm{ps}}^\R(S^1)$ be the corresponding central extensions by $\R$, constructed by using the same cocycle that was used to construct the central extensions of $\Diff_0(I)$ and of~$\Diff_*(I)$.

By \cite[App.\,A]{arXiv:1811.04501}, the maps $\Diff_0^\R(I) \to \cA_{G,k}(I)$ extend to the larger group $\Diff_{1,\mathrm{ps}}^\R(I)$.
The proofs in Section \ref{sec:Diffeomorphism groups} go through with the groups 
$\Diff_{1,\mathrm{ps}}(I)$, $\Diff_{1,\mathrm{ps}}^\R(I)$ and $\Diff_{1,\mathrm{ps}}^\R(S^1)$ in place of $\Diff_0(I)$, $\Diff_0^\R(I)$ and $\Diff_*^\R(S^1)$.
In particular, the homomorphism~\eqref{eq: assembled action + Diff} extends to a homomorphism
\[
\Diff_{1,\mathrm{ps}}^\R(S^1)\ltimes \widetilde{\Omega G}\,\to\, U(H).
\]
Let $\R\subset \Diff_{1,\mathrm{ps}}^\R(S^1)$ be (the canonical lift of) the subgroup of M\"obius transformations that fixes~$p$.
Upon mapping $S^1$ to the real line via the stereographic projection that sends $p=1$ to $\infty$, this group gets identified with the group of translations of the real line.
We write $P$ for the infinitesimal generator, and call it the energy-momentum operator
(if the Hilbert space has an action of $\Diff^\R(S^1)$, then the energy-momentum operator is given by $P=-L_{-1}+2L_0-L_1$).

\begin{conjecture}\label{conj: weaker conjecture}
For every locally normal representation $H\in\Rep^k_{\scriptscriptstyle \mathrm{l.n.\!\!}}(\Omega G)$ of the based loop group, the energy-momentum operator $P$ has positive spectrum.
\end{conjecture}

The above conjecture has been recently proven by Del Vecchio, Iovieno, and Tanimoto \cite[Thm\,3.4]{arXiv:1811.04501} (and is thus no longer a conjecture).

We define a \emph{positive energy representation} of the based loop group to be a representation whose energy-momentum operator has positive spectrum:

\begin{definition} 
A level $k$ positive energy representation of the based loop group is
a continuous representation $\rho:\widetilde{\Omega G}\to U(H)$ satisfying:
\begin{enumerate}
\item $\rho$ is the restriction of a representation $\R\ltimes \widetilde{\Omega G}\to U(H)$ such that the infinitesimal generator $P$ of the group $\R$ has positive spectrum.
\item $\lambda\in U(1)\subset \widetilde{\Omega G}$ acts by scalar multiplication by $\lambda^k$.
\end{enumerate}
We note that the action of $\R$ is not part of the data of a positive energy representation.
\end{definition}

The following is a strengthening of Conjecture~\ref{conj: weaker conjecture}:

\begin{conjecture}\label{conj: stronger conjecture}
A representation of the centrally extended based loop group is locally normal 
if and only if it has positive energy.
\end{conjecture}

\bibliographystyle{amsalpha}

\begin{thebibliography}{CDVIT18}

\bibitem[BDH15]{BDH-cn1}
Arthur Bartels, Christopher~L. Douglas, and Andr{\'e} Henriques,
  \emph{Conformal nets {I}: {C}oordinate-free nets}, Int. Math. Res. Not. IMRN
  (2015), no.~13, 4975--5052. \MR{3439097}

\bibitem[BE98]{MR1652746}
J.~B{\"o}ckenhauer and D.~E. Evans, \emph{Modular invariants, graphs and
  {$\alpha$}-induction for nets of subfactors. {I}}, Comm. Math. Phys.
  \textbf{197} (1998), no.~2, 361--386. \MR{1652746 (2000c:46121)}

\bibitem[Bot77]{MR0488080}
Raoul Bott, \emph{On the characteristic classes of groups of diffeomorphisms},
  Enseignment Math. (2) \textbf{23} (1977), no.~3-4, 209--220. \MR{0488080 (58
  \#7651)}

\bibitem[BSM90]{MR1079298}
Detlev Buchholz and Hanns Schulz-Mirbach, \emph{Haag duality in conformal
  quantum field theory}, Rev. Math. Phys. \textbf{2} (1990), no.~1, 105--125.
  \MR{1079298 (92a:81106)}

\bibitem[Car04]{MR2031030}
Sebastiano Carpi, \emph{On the representation theory of {V}irasoro nets}, Comm.
  Math. Phys. \textbf{244} (2004), no.~2, 261--284. \MR{2031030}

\bibitem[CDVIT18]{arXiv:1808.02384}
Sebastiano Carpi, Simone Del~Vecchio, Stefano Iovieno, and Yoh Tanimoto,
  \emph{Positive energy representations of sobolev diffeomorphism groups of the
  circle}, arXiv:1808.02384, 2018.

\bibitem[CW05]{MR2166846}
Sebastiano Carpi and Mih\'{a}ly Weiner, \emph{On the uniqueness of
  diffeomorphism symmetry in conformal field theory}, Comm. Math. Phys.
  \textbf{258} (2005), no.~1, 203--221. \MR{2166846}

\bibitem[CW16]{Local_energy_bounds}
Sebastiano Carpi and Mih\'aly Weiner, \emph{Local energy bounds and
  representations of conformal nets}, unpublished, 2016.

\bibitem[DFK04]{MR2078164}
Claudio D'Antoni, Klaus Fredenhagen, and S\o~ren K\"{o}ster,
  \emph{Implementation of conformal covariance by diffeomorphism symmetry},
  Lett. Math. Phys. \textbf{67} (2004), no.~3, 239--247. \MR{2078164}

\bibitem[DVIT18]{arXiv:1811.04501}
Simone Del~Vecchio, Stefano Iovieno, and Yoh Tanimoto, \emph{Solitons and
  nonsmooth diffeomorphisms in conformal nets}, arXiv:1811.04501, 2018.

\bibitem[GF68]{MR0245035}
I.~M. Gel'fand and D.~B. Fuks, \emph{Cohomologies of the {L}ie algebra of
  vector fields on the circle}, Funkcional. Anal. i Prilo\v zen. \textbf{2}
  (1968), no.~4, 92--93. \MR{0245035}

\bibitem[GF93]{MR1231644}
Fabrizio Gabbiani and J{\"u}rg Fr{\"o}hlich, \emph{Operator algebras and
  conformal field theory}, Comm. Math. Phys. \textbf{155} (1993), no.~3,
  569--640. \MR{1231644 (94m:81090)}

\bibitem[GW84]{MR733047}
Roe Goodman and Nolan~R. Wallach, \emph{Structure and unitary cocycle
  representations of loop groups and the group of diffeomorphisms of the
  circle}, J. Reine Angew. Math. \textbf{347} (1984), 69--133. \MR{733047
  (86g:22024a)}

\bibitem[Hen15]{CS(pt)}
Andr{\'e} Henriques, \emph{What {C}hern--{S}imons theory assigns to a point},
  arXiv:1503.06254, 2015.

\bibitem[Hen16]{WZW-classification}
\bysame, \emph{The classification of chiral {WZW} models by
  ${H}^4_+({B}{G},\mathbb {Z})$}, Contemporary Mathematics (2016).

\bibitem[Hen17]{Bicommutant-categories-from-conformal-nets}
\bysame, \emph{Bicommutant categories from conformal nets}, arXiv:1701.02052,
  2017.

\bibitem[Kac90]{MR1104219}
Victor~G. Kac, \emph{Infinite-dimensional {L}ie algebras}, third ed., Cambridge
  University Press, Cambridge, 1990. \MR{1104219 (92k:17038)}

\bibitem[Kaw02]{MR1892455}
Yasuyuki Kawahigashi, \emph{Generalized {L}ongo-{R}ehren subfactors and
  {$\alpha$}-induction}, Comm. Math. Phys. \textbf{226} (2002), no.~2,
  269--287. \MR{1892455 (2003i:46067)}

\bibitem[KLM01]{MR1838752}
Yasuyuki Kawahigashi, Roberto Longo, and Michael M{\"u}ger,
  \emph{Multi-interval subfactors and modularity of representations in
  conformal field theory}, Comm. Math. Phys. \textbf{219} (2001), no.~3,
  631--669. \MR{1838752 (2002g:81059)}

\bibitem[KW09]{MR2456522}
Boris Khesin and Robert Wendt, \emph{The geometry of infinite-dimensional
  groups}, Ergebnisse der Mathematik und ihrer Grenzgebiete. 3. Folge. A Series
  of Modern Surveys in Mathematics [Results in Mathematics and Related Areas.
  3rd Series. A Series of Modern Surveys in Mathematics], vol.~51,
  Springer-Verlag, Berlin, 2009. \MR{2456522 (2009k:58016)}

\bibitem[Lok94]{Loke(PhD-thesis)}
Terence Loke, \emph{Operator algebras and conformal field theory of the
  discrete series representations of $\mathit{Diff}(s^1)$}, Ph.D. thesis,
  Trinity College, Cambridge (1994).

\bibitem[LR95]{MR1332979}
R.~Longo and K.-H. Rehren, \emph{Nets of subfactors}, Rev. Math. Phys.
  \textbf{7} (1995), no.~4, 567--597, Workshop on Algebraic Quantum Field
  Theory and Jones Theory (Berlin, 1994). \MR{1332979 (96g:81151)}

\bibitem[LX04]{MR2100058}
Roberto Longo and Feng Xu, \emph{Topological sectors and a dichotomy in
  conformal field theory}, Comm. Math. Phys. \textbf{251} (2004), no.~2,
  321--364. \MR{2100058 (2005i:81087)}

\bibitem[Mic89]{MR1032521}
Jouko Mickelsson, \emph{Current algebras and groups}, Plenum Monographs in
  Nonlinear Physics, Plenum Press, New York, 1989. \MR{1032521}

\bibitem[Ol{\cprime}91]{MR1191619}
A.~Yu. Ol{\cprime}shanski{\u\i}, \emph{Geometry of defining relations in
  groups}, Mathematics and its Applications (Soviet Series), vol.~70, Kluwer
  Academic Publishers Group, Dordrecht, 1991, Translated from the 1989 Russian
  original by Yu. A. Bakhturin. \MR{1191619 (93g:20071)}

\bibitem[PS86]{MR900587}
Andrew Pressley and Graeme Segal, \emph{Loop groups}, Oxford Mathematical
  Monographs, The Clarendon Press, Oxford University Press, New York, 1986,
  Oxford Science Publications. \MR{900587 (88i:22049)}

\bibitem[TL97]{Toledano(PhD-thesis)}
Valerio Toledano~Laredo, \emph{Fusion of positive energy representations of
  ${LS}pin_{2n}$}, Ph.D. thesis, St. John's College, Cambridge (1997).

\bibitem[TL99]{MR1674631}
\bysame, \emph{Integrating unitary representations of infinite-dimensional
  {L}ie groups}, J. Funct. Anal. \textbf{161} (1999), no.~2, 478--508.
  \MR{1674631 (2000f:22023)}

\bibitem[Tsu02]{MR1882783}
Takashi Tsuboi, \emph{On the perfectness of groups of diffeomorphisms of the
  interval tangent to the identity at the endpoints}, Foliations: geometry and
  dynamics ({W}arsaw, 2000), World Sci. Publ., River Edge, NJ, 2002,
  pp.~421--440. \MR{1882783}

\bibitem[Was90]{Wass_unpub_1990}
Antony Wassermann, \emph{Subfactors arising from positive energy
  representations of some infinite-dimensional groups}, unpublished, 1990?

\bibitem[Was98]{MR1645078}
\bysame, \emph{Operator algebras and conformal field theory. {III}. {F}usion of
  positive energy representations of {${\rm LSU}(N)$} using bounded operators},
  Invent. Math. \textbf{133} (1998), no.~3, 467--538. \MR{1645078 (99j:81101)}

\bibitem[Wei06]{MR2231680}
Mih\'aly Weiner, \emph{Conformal covariance and positivity of energy in charged
  sectors}, Comm. Math. Phys. \textbf{265} (2006), no.~2, 493--506.
  \MR{2231680}

\bibitem[Wei17]{MR3627412}
Mih\'{a}ly Weiner, \emph{Local equivalence of representations of {${\rm
  Diff}^+(S^1)$} corresponding to different highest weights}, Comm. Math. Phys.
  \textbf{352} (2017), no.~2, 759--772. \MR{3627412}

\bibitem[Xu00]{MR1776984}
Feng Xu, \emph{Jones-{W}assermann subfactors for disconnected intervals},
  Commun. Contemp. Math. \textbf{2} (2000), no.~3, 307--347. \MR{1776984
  (2001f:46094)}

\bibitem[Zel15]{Zellner}
Christoph Zellner, \emph{On the existence of regular vectors},
  arXiv:1510.08727, 2015.

\end{thebibliography}

\def\cprime{$'$} \def\cprime{$'$} \def\cprime{$'$} \def\cprime{$'$}
\providecommand{\bysame}{\leavevmode\hbox to3em{\hrulefill}\thinspace}
\providecommand{\MR}{\relax\ifhmode\unskip\space\fi MR }
\providecommand{\MRhref}[2]{%
  \href{http://www.ams.org/mathscinet-getitem?mr=#1}{#2}
}
\providecommand{\href}[2]{#2}

\end{document}